\documentclass[journal]{IEEEtran}
\usepackage[T1]{fontenc}
\usepackage[latin9]{inputenc}
\usepackage{color}
\usepackage{float}
\usepackage{amsthm}
\usepackage[nosumlimits]{amsmath}
\usepackage{amssymb}
\usepackage{multirow}
\usepackage{overpic}

\makeatletter

\floatstyle{ruled}
\newfloat{algorithm}{tbp}{loa}
\providecommand{\algorithmname}{Algorithm}
\floatname{algorithm}{\protect\algorithmname}

\theoremstyle{plain}

\theoremstyle{definition}

\theoremstyle{plain}

\theoremstyle{remark}
\newtheorem{rem}{\protect\remarkname}
\theoremstyle{remark}
\newtheorem{lm}{\protect\lemmaname}

\allowdisplaybreaks

\usepackage{mathtools}

\usepackage{amsfonts}
\usepackage{cite}
\usepackage{array}
\usepackage{algorithm}
\usepackage{algcompatible}

\usepackage{graphicx}
\usepackage{subfigure}

\DeclareMathOperator*{\maxi}{max}

\DeclareMathOperator*{\st}{s.t.}

\newcommand{\herm}{^{H}}
\newcommand{\trans}{^{T}}

\renewcommand{\Re}{\mathrm{Re}}

\newcommand*{\rom}[1]{\expandafter\@slowromancap\romannumeral #1@}

\makeatother

\providecommand{\definitionname}{Definition}
\providecommand{\factname}{Fact}
\providecommand{\remarkname}{Remark}
\providecommand{\theoremname}{Theorem}
\providecommand{\lemmaname}{Lemma}

\mathchardef\mhyphen="2D

\IEEEoverridecommandlockouts

\begin{document}
\title{Energy-Efficient Multi-Cell Multigroup Multicasting with Joint Beamforming and Antenna Selection}
%\author{Oskari Tervo, Le-Nam Tran, Harri Pennanen, Symeon Chatzinotas, Markku Juntti, and Bj\"orn Ottersten
\author{Oskari~Tervo,~\IEEEmembership{Student~Member,~IEEE,}
        Le-Nam Tran,~\IEEEmembership{Member,~IEEE,}
        Harri Pennanen,~\IEEEmembership{Member,~IEEE,}
        Symeon Chatzinotas,~\IEEEmembership{Senior Member,~IEEE,}
        Bj\"orn Ottersten,~\IEEEmembership{Fellow,~IEEE,}
        and~Markku~Juntti,~\IEEEmembership{Senior~Member,~IEEE}%% <-this % stops a space
%\author{\IEEEauthorblockN{Oskari Tervo\textsuperscript{$\ast$}, Le-Nam Tran\textsuperscript{$\dag$}, Harri Pennanen\textsuperscript{$\ast$}, Symeon Chatzinotas\textsuperscript{$\ddagger$}, Markku Juntti\textsuperscript{$\ast$}, and Bj\"orn Ottersten\textsuperscript{$\ddagger$}}
%\IEEEauthorblockA{\textsuperscript{$\ast$}Centre for Wireless Communications, University of Oulu, Oulu, Finland\\
%\textsuperscript{$\dag$}Department of Electronic Engineering, Maynooth University, Maynooth, Co Kildare, Ireland\\
%\textsuperscript{$\ddagger$}the Interdisciplinary Centre for Security, Reliability and Trust, University of Luxembourg, Luxembourg \\
%Email: \{oskari.tervo,harri.pennanen,markku.juntti\}@oulu.fi, lenam.tran@nuim.ie, \{symeon.chatzinotas, bjorn.ottersten\}@uni.lu}
\thanks{This research has been financially supported by Academy of Finland 6Genesis Flagship (grant 318927). It was also supported in part by Infotech Oulu Doctoral Program and the Academy of Finland under projects MESIC belonging to the WiFIUS program with NSF, and WiConIE. It has also emanated from research supported in part by a Grant from Science Foundation Ireland under Grant number 17/CDA/4786. It was further supported by projects PROSAT, SATSENT, INWIPNET and H2020 SANSA. The first author has been supported by Oulu University Scholarship Foundation, Nokia Foundation, Tauno T{\"o}nning Foundation, and Walter Ahlstr{\"o}m Foundation.}
\thanks{O. Tervo was with Centre for Wireless Communications, University of Oulu, Finland, and is now with Nokia Bell Labs, Finland. Email: oskari.tervo@nokia-bell-labs.com.}
\thanks{H. Pennanen and M. Juntti are with Centre for Wireless Communications, University of Oulu, Finland. Email: \{harri.pennanen,markku.juntti\}@oulu.fi.}
\thanks{L.-N. Tran is with the School of Electrical and Electronic Engineering, University College Dublin, Ireland. Email: nam.tran@ucd.ie.}
\thanks{S. Chatzinotas and B. Ottersten are with the Interdisciplinary Centre for Security, Reliability and Trust, University of Luxembourg, Luxembourg. Email: \{symeon.chatzinotas,bjorn.ottersten\}@uni.lu.}
%\thanks{Part of this paper has been presented at IEEE International Conference on Communications, Paris, France, 2017.}
}

%\author{Oskari Tervo, Le-Nam Tran, Harri Pennanen, Symeon Chatzinotas, Markku Juntti, and Bj\"orn Ottersten
%%\author{Oskari~Tervo,~\IEEEmembership{Student~Member,~IEEE,}
%%        Antti~Tï¿½lli,~\IEEEmembership{Member,~IEEE,}
%%        Le Nam~Tran,~\IEEEmembership{Member,~IEEE,}
%%        and~Markku~Juntti,~\IEEEmembership{Senior~Member,~IEEE}% <-this % stops a space
%\thanks{This work was supported in part by Infotech Oulu Doctoral Program and the Academy of Finland under project Message and CSI Sharing for Cellular Interference Management with Backhaul Constraints (MESIC) belonging to the WiFIUS program with NSF. It has
%also been co-funded by the Irish Government and the European Union under Irelands EU Structural and Investment Funds Programmes 2014-2020 through the SFI Research Centres Programme under Grant 13/RC/2077.}
%\thanks{O. Tervo, A. Tölli and Markku Juntti are with Centre for Wireless Communications, University of Oulu,
%Finland. Email: \{oskarite, antti.tolli, markku.juntti\}@ee.oulu.fi.}
%\thanks{L.-N. Tran is with the Department of Electronic Engineering, Maynooth University, Ireland. Email: lenam.tran@nuim.ie.}}
%Finland. Email: \{oskarite, ltran, markku.juntti\}@ee.oulu.fi. This research
%is supported by the Academy of Finland, the Finnish Funding Agency for Technology and Innovation and InfoTech Oulu Doctoral Program.}}
%Further supporters have been Tauno Tönning Foundation, Walter Ahlström Foundation, KAUTE foundation, and Finnish Foundation for Technology Promotion.

\maketitle

\begin{abstract}
This paper studies the energy efficiency and sum rate trade-off for coordinated beamforming in multi-cell multi-user multigroup multicast multiple-input single-output systems. We first consider a conventional network energy efficiency maximization (EEmax) problem by jointly optimizing the transmit beamformers and antennas selected to be used in transmission. We also account for  per-antenna maximum power constraints to avoid non-linear distortion in power amplifiers and user-specific minimum rate constraints to guarantee certain service levels and fairness. To be energy-efficient, transmit antenna selection is employed. It eventually leads to a  mixed-Boolean fractional program.  We then propose two different approaches to solve this difficult problem. The first solution is based on a novel modeling  technique that produces a tight continuous relaxation.
% * <namletran@gmail.com> 2017-09-27T12:22:14.466Z:
%
% > Boolean
% Try to be consistent in the paper. Choose Boolean or binary and stick with that
%
% ^.
%of binary variables, and we propose a formulation which converges very close to a binary solution. Thus, we can simply switch off those antennas for which the antenna selection variables are zero, and use the resulting beamformers for transmission.
The second approach is based on sparsity-inducing method, which does not require the introduction of any Boolean variable. We also investigate the trade-off between the energy efficiency and sum rate  by proposing two different formulations. In the first formulation, we propose a new metric that is the ratio of the sum rate and the so-called weighted power. Specifically, this metric reduces to EEmax when the weight is 1, and to sum rate maximization when the weight is 0. In the other method, we treat the trade-off problem as  a multi-objective optimization for which a scalarization approach is adopted. Numerical results illustrate significant achievable energy efficiency gains over the method where the antenna selection is not employed. The effect of antenna selection on the energy efficiency and sum rate trade-off is also  demonstrated.

% We first introduce binary antenna selection variables and use the perspective formulation to model the relation between them and the beamformers. Subsequently, we propose a new formulation which reduces the feasible set of the continuous relaxation, resulting in better performance compared to the original perspective formulation based problem. However, the resulting optimization problem is a mixed-Boolean non-convex fractional program, which is difficult to solve. We follow the standard continuous relaxation of the binary antenna selection variables, and then reformulate the problem such that it is amendable to successive convex approximation. Thereby, solving the continuous relaxation mostly results in near-binary solution. To recover the binary variables from the continuous relaxation, we switch off all the antennas for which the continuous values are smaller than a small threshold. Numerical results illustrate the superior convergence result and significant achievable gains in terms of energy efficiency with the proposed algorithm. The effect of antenna selection on the energy and spectral efficiency trade-off is demonstrated.

\end{abstract}

\begin{IEEEkeywords}
Coordinated beamforming, energy efficiency, successive convex approximation, fractional programming, antenna selection, multicasting, multi-objective optimization.
\end{IEEEkeywords}

\section{\label{sec:intro}Introduction}

Achieving high energy efficiency (EE) and spectral efficiency (SE) is vital to future wireless communications standards. SE maximization has  driven  cellular networks to employ   aggressive frequency reuse. Essentially, different base stations (BSs) transmit data in the same frequency spectrum, resulting in severe inter-user interference conditions. In this case, multi-antenna system can exploit beamforming to control the interference for efficient spectrum utilization. An efficient method in this regard is coordinated beamforming \cite{Irmer-11}, where the base stations design the beams in a coordinated manner.

\textcolor{black}{Previous studies have shown that EE and SE are conflicting targets \cite{Xiong-11,He-13EESE,Meshkati-07,Chen-11,Tervo-15}, if the power consumption due to additional hardware caused by increasing the number of antennas is taken into account \cite{FoundAndTrends-17}. More specifically, it may be energy-efficient to transmit with a small number of antennas if the power cost due to an active antenna is large, and, thus, the spectral efficiency can be low \cite{Tervo-15}. On the other hand, the SE maximization requires that  base stations are equipped with a large number of antennas to avail of spatial diversity. This increases the network power consumption and starts to reduce EE when the cost due to the power consumption increase exceeds the benefit of the SE increase \cite{Bjornson-15}.} \textcolor{black}{If the number of antennas is fixed and we wish to use conventional digital beamforming, then we could not adjust the RF chain power consumption and the EE-SE could be adjusted only by changing the beamformers (i.e., changing the transmit power as a result). On the other hand, a good option to trade-off the two design targets is to use antenna selection techniques. Specifically, depending on the required data rates, one could switch off some RF chains to save power. In this regard, we could generally install a large number of antennas, and then use a proper antenna selection scheme to control EE-SE trade-off together with beamforming. Consequently, we can achieve a better trade-off  compared to the conventional method because both RF chain and transmit powers can be adjusted. The antenna selection reduces both the power consumption and the SE, but when it is optimized for EE together with beamforming, one achieves ideally increasing EE as a function of the number of antennas.} This idea motivates the joint optimization of both transmit beamformers and active transmit antennas \cite{Tervo-15,Mehanna-13}.
%On the other hand, for the spectrally efficient communications, the best strategy is to use all the possible spatial dimensions.
In practice, both EE and SE performance measures are important for mobile network operators, depending on the user distribution and service requirements. To this end, the energy and spectral efficiency trade-off problem has been considered in the recent literature \cite{Xiong-11,He-13EESE}.

The evolution of mobile handsets and the associated applications is creating a new type of wireless communication scenario. A large part of the requested data traffic from users is highly correlated, especially in crowded areas, e.g., in stadiums. To deal with such situations, multicasting has received special  attention as a promising solution \cite{Vazquez-16,Pen-16Sat,Sidiropoulos-06,Karipidis-08,Christopoulos-14,FrameBased-15,Mehanna-13,Xiang-13,He-15,Tervo-17MIMO,Tervo-17MinPower}. The idea is to transmit the same information to multiple users as a single transmission, and it has become increasingly popular in the context of cache-enabled cloud radio access networks (C-RANs) proposed for 5G systems to improve both spectral and energy efficiency \cite{Tao-15}.

\subsection{Related Work}

%The research presented before focuses mostly on unicast beamforming where each user is assigned with an independent data stream. However, the increasing demand for high data rate applications such as video broadcasting services creates new challenges. This gives rise to the concept of multicast beamforming, where multiple users desire to receive the same information. Multicasting is a particularly powerful technique in the context of cache-enabled cloud radio access networks proposed for 5G systems, where it can be used to transmit the same popular contents to multiple users to improve both spectral and energy efficiency \cite{Tao-15}.
% The physical layer multicasting has been also included in the LTE standards and it has applications in satellite communications \cite{Vazquez-16,Pen-16Sat} as well.

Energy and spectral efficiency trade-off problems have been studied in different works. \textcolor{black}{In \cite{Xiong-11}, fundamental EE-SE trade-offs were studied for joint power and subcarrier allocation in a single-cell single-input single-output (SISO) downlink orthogonal frequency-division multiple access (OFDMA) system. A distributed antenna  system (DAS) with single-antenna nodes was considered in \cite{He-13EESE}, where a weighted sum method was proposed to solve the multi-criteria optimization problem. In \cite{Tang-16}, joint beamforming and subcarrier allocation for single-cell SISO downlink systems was studied. A weighted sum approach was proposed to the trade-off problem in terms of resource efficiency, which involves a normalization factor to balance the values of EE and SE. A single-cell OFDMA system with imperfect CSI was considered in \cite{Amin-16}. In \cite{Tang-14}, the EE-SE trade-off was investigated in a single-cell  multiple-input multiple-output (MIMO) OFDMA system and the authors considered non-linear dirty paper coding (DPC) with antenna and subcarrier selection.} However, all these previous studies focus on unicasting and mostly SISO transmission, where each user is assigned an independent data stream. Although \cite{Tang-14} focused on a MIMO case, the use of DPC makes it difficult to implement in reality.

Beamforming design for multicasting  has been studied for single-cell systems for different optimization targets, e.g., transmit power minimization \cite{Sidiropoulos-06,Karipidis-08,Tervo-17MinPower}, max-min fairness \cite{Karipidis-08,Christopoulos-14,Tervo-17MinPower}, and sum rate maximization \cite{FrameBased-15}. Joint beamforming and antenna selection for transmit power minimization was studied in \cite{Mehanna-13}.
Coordinated multicast beamforming for transmit power minimization and max-min fairness has been studied in \cite{Xiang-13}. In \cite{Tervo-17MIMO}, energy-efficient joint unicasting and multicasting beamforming for multi-cell multi-user MIMO systems was considered. A method to solve the EE maximization problem in multi-cell system with single group per cell was proposed in \cite{He-15}. However, both \cite{He-15} and \cite{Tervo-17MIMO} only considered the beamforming problem without taking into account the fact that significant energy  savings can be achieved by switching off some of the RF chains, i.e., antenna selection. Moreover, the works of \cite{He-15,Tervo-17MIMO} only considered the case of sum power constraints, while the case of antenna-specific power constraints has to be handled differently.

%The above works have assumed the conventional beamforming approach where all the interference is treated as noise. Joudeh {\it et al.} \cite{Joudeh-16} proposed to use different approach based on rate-splitting. They considered a single-cell multigroup system where each group message is divided into a private and common part. All the common parts are then encoded into a single stream and decoded by all the user groups, while the private parts are encoded as conventionally. They illustrated significant performance gains over the conventional beamforming in terms of max-min fairness in the interference-limited setups. However, applying the method to a multi-cell system is more challenging because a single beamformer for the common stream as in \cite{Joudeh-16} cannot be realized due to the fact that different user groups are served by different BSs.
%\textcolor{blue}{The ADMM-consensus approach for general quadratically constrained quadratic programs was considered in \cite{Huang-16} which can be applied to single-cell multicasting.}

\subsection{Contributions}

In this paper, we study energy-efficient coordinated beamforming in multi-cell multigroup multi-user multicast multiple-input single-output (MISO) systems. Each transmit antenna is subject to an individual maximum power constraint and each user is guaranteed with a minimum data rate. We focus on a case where the number of antennas is relatively large, so that there is a potential to switch off some of the transmit antennas to improve the energy efficiency. In this setup, we consider the joint optimization of beamforming and antenna selection, where novel and clever formulations and transformations are proposed so that widely used standard optimization techniques can be applied to solve the problem efficiently.
Specifically, two different approaches are proposed. In the first one, we introduce Boolean antenna selection variables and use a novel extension of the perspective formulation \cite{Gunluk-10,MINLP} to model the per antenna power constraints. In particular, a specific parameter is introduced to control the  tightness of the  continuous relaxation which is crucial to finding a high-quality feasible solution. Since the continuous relaxation is nonconvex, we propose a successive convex approximation (SCA) based algorithm to solve it. By novel transformations, the subproblems obtained at each iteration of the proposed method can be approximated as a second-order cone program (SOCP) for which modern convex solvers are particularly efficient.
%Thanks to the proposed formulation, the continuous relaxation yields some continuous values being close to binary ones. To this end, we switch off antennas which have continuous values close to zero.
The second direction is based on a sparse beamforming approach where the idea is to directly find sparse beamforming solutions without requiring any additional variables compared to the original beamforming design problem without antenna selection. We propose different convex and non-convex smoothing functions to approximate the $\ell_0$-`norm', and again employ SCA to solve the problem.
% that does not increase the problem size, compared to the original beamforming design problem without antenna selection.
The numerical results are provided to illustrate the convergence of the proposed algorithms for different system parameters and the achieved energy efficiency gains using the joint beamforming and antenna selection.

In the second part, we extend the joint design to the energy efficiency and sum rate trade-off  problem. In this case, the considered joint beamforming and antenna selection problem is specially relevant, because such a design certainly achieves a better trade-off curve due to extra degrees of freedom provided by the antenna selection. To formulate this multi-objective optimization problem, we propose two different approaches. First, we propose a new optimization metric, namely the power-weighted energy efficiency (PWEE) maximization, which involves a weighting parameter for the adjustable power consumption. The benefit of this approach is that the  algorithms  derived for EE maximization can be straightforwardly used to solve the PWEE problem. The other approach is attained via scalarization  where the sum rate function is appropriately scaled to achieve practical trade-off for the weighted sum of energy efficiency and sum rate. Due to the more difficult structure of the objective function, another set of approximated constraints is required compared to the EEmax problem to solve the problem. Numerical results demonstrate that both designs can exploit the trade-off and that joint beamforming and antenna selection achieves significantly wider trade-off curve compared to the case where only beamforming design is exploited.

\subsection{Organization and Notation}
The rest of the paper is organized as follows. Section \ref{sec:ProblemFormulation} presents the system model, power consumption model and the EE maximization problem. The proposed EE maximizing algorithms are provided in Sections \ref{sec:CentralizedMethods} and \ref{sec:Sparse}, while the trade-off problem is studied in  Section \ref{sec:tradeoff}. The numerical results and conclusions are presented in Sections \ref{sec:NumericalResults} and \ref{sec:Conclusions}, respectively.

The following notations are used in this paper. We denote by $|x|$ the cardinality of $x$ if $x$ is a set, and absolute value of $x$, otherwise. The $i$th component of vector $\mathbf{x}$ is denoted by $\mathbf{x}[i]$. Notation $||\mathbf{x}||_2$ is the Euclidean norm of $\mathbf{x}$, boldcase letters are vectors, $\mathbf{x}\trans, \mathbf{x}\herm, \Re(\mathbf{x})$ mean transpose, Hermitian transpose, and real part of $\mathbf{x}$, respectively.  For a positive integer $K$, $\mathcal{K}$ is defined as the set $\{1,\ldots,K\}$.%\vspace{-1mm}

\section{\label{PF} System Model and Problem Formulation}
\label{sec:ProblemFormulation}
\subsection{System Model} %\vspace{-2mm}
A multi-cell multigroup multicasting system consisting of $B$ BSs is considered, \textcolor{black}{as illustrated in Fig. \ref{fig: MultiCastingSystemThesis}}. Each BS $b \in \mathcal{B}$ equipped with \textcolor{black}{$N_b=|\mathcal{N}_b| (\mathcal{N}_b=\{1,\ldots,N_b\})$} antennas, has $G_b$ multicasting groups to serve, that is, each group desires to receive independent information from its serving BS. The set of groups served by BS $b$ is denoted by $\mathcal{G}_b \subset \mathcal{G}$, where $\mathcal{G}$ is the set of all groups in the network. The total number of single-antenna users in the network is denoted by $K=|\mathcal{K}| (\mathcal{K}=\{1,\ldots,K\})$, while the user set belonging to group $g$ is denoted by $\mathcal{K}_g \subset \mathcal{K}$.
% * <namletran@gmail.com> 2017-09-27T15:55:34.027Z:
%
% >  $\mathcal{G}_b=\{j,j+1,\ldots,j+G_b-1\}$
% This is ill defined.  Try to make all the notations clear.
%
% ^.
The serving BS of user group $g$ is denoted as $b_g$.
%The set of users in group $g$ is denoted by $\mathcal{K}_g \subset \mathcal{K}$.
The sets of users belonging to different groups are disjoint, i.e., $\mathcal{K}_i \cap \mathcal{K}_j = \emptyset$, $\forall i,j \in \mathcal{G}, i \neq j$. In other words, each user is assumed to belong to one group only. Each group is further served by one BS only.
User $k$ in group $g$ receives the signal
\textcolor{black}{\begin{eqnarray} \label{eq:RxSignal}
{y}_{k} &=& \overbrace{\mathbf{h}_{b_g,k}\herm \mathbf{F}_{b_g} \mathbf{w}_{g} {s}_{g}}^\text{desired signal} + \overbrace{\sum\limits_{i \in \mathcal{G}_{b} \setminus \{g\}} \mathbf{h}_{b_g,k}\herm \mathbf{F}_{b_g} \mathbf{w}_{i} {s}_{i}}^\text{\textcolor{black}{inter-group interference from the same cell}} \nonumber \\
&& + \underbrace{\sum\limits_{j \in \mathcal{B} \setminus \{b_g\}} \sum\limits_{u \in \mathcal{G}_j} \mathbf{h}_{j,k}\herm \mathbf{F}_{b_u}\mathbf{w}_{u} {s}_{u}}_\text{\textcolor{black}{inter-group interference from the other cells}} + {n}_{k}
\end{eqnarray}
where $\mathbf{h}_{b,k} \in \mathbb{C}^{N_b \times 1}$ is the channel vector from BS $b$ to user $k$, $\mathbf{w}_{g} \in \mathbb{C}^{N_b \times 1}$ is the transmit beamforming vector of group $g$, ${s}_{g} \in \mathbb{C}$ is the corresponding independent normalized data symbol, ${n}_{k} \ \sim \mathcal{C} \mathcal{N} (0, \sigma^{2})$ is the complex white Gaussian noise sample with zero mean and variance $\sigma^{2}$,\footnote{\textcolor{black}{The noise variance is assumed to be same for all the users without loss of generality.}} and $\mathbf{F}_{b_u} \in \mathbb{R}^{N_b\times N_b}$ is the antenna selection matrix involving the $i$th unit vector at the $i$th column if the $i$th antenna is selected and otherwise a zero vector.} The channel vectors are assumed to be perfectly known at the transmitters, while the receivers are assumed to have perfect effective channel information to decode the data. The multigroup interference is treated as Gaussian noise, yielding the SINR of user $k$  as
% * <namletran@gmail.com> 2017-09-27T15:57:43.451Z:
%
% > SINR
% Mention that interference is simply treated as Gaussian noise.
%
% ^.
\begin{equation}\label{eq:SINRexpression}
\hat{\Gamma}_k(\mathbf{w}) = \frac{|\mathbf{h}_{b_g,k}\herm \mathbf{F}_{b_g}\mathbf{w}_{g}|^{2}}{{N_0 + \sum\limits_{u \in \mathcal{G} \setminus \{g\}} |\mathbf{h}_{b_u,k}\herm \mathbf{F}_{b_u}\mathbf{w}_{u}|^{2}}}
\end{equation}
where $N_0$ is the total noise power over the transmission bandwidth $W$, \textcolor{black}{and $\mathbf{w}\triangleq\{\mathbf{w}_g\}_{g\in\mathcal{G}}$}.
As a result, the data rate towards user $k$ is given as
\begin{equation}
R_k(\mathbf{w}) \triangleq W\log(1+\hat{\Gamma}_k(\mathbf{w})).\footnote{Since the transmission bandwidth is fixed throughout the paper, it is discarded in the mathematical derivations for notational simplicity.}
\end{equation}

\begin{figure}[t]
\centering
  \includegraphics[width=0.7\columnwidth]{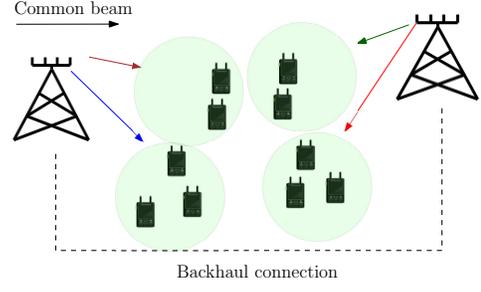}
  \caption{A multi-cell multigroup multicasting system.}
  \label{fig: MultiCastingSystemThesis}
\end{figure}

\subsection{Power Consumption Model}

In this paper the total power consumption is modeled as \textcolor{black}{\cite{Bjornson-15}}
\begin{equation}
\label{eq:powermodel}
\begin{aligned}
\hat{P}_{\text{tot}}= & \dfrac{1}{\eta}\sum\limits_{g\in\mathcal{G}}||\mathbf{F}_{b_g}\mathbf{w}_{g}||_2^{2} \\
&  +P_{\text{RF}}\sum\limits_{b\in\mathcal{B}}\sum\limits_{i \in \mathcal{N}_b}\mathbf{F}_{b}[i,i] + BP_{\text{sta}} + |\mathcal{K}|P_{\text{UE}}
\end{aligned}
\end{equation}
where the first term is the PAs' power consumption to get the desired output powers assuming PA efficiency $\eta\in[0,1]$. The second term is the power consumption of the RF chains, i.e., \textcolor{black}{an amount of  $P_{\text{RF}}$ is consumed if the $i$th antenna of BS $b$ is selected and there is no power consumption otherwise}. $P_{\text{sta}}$
is the static power spent by cooling systems, power supplies, etc, and $P_{\text{UE}}$ is the power consumption of each user terminal. For the ease of notation, we denote $P_0\triangleq BP_{\text{sta}} + |\mathcal{K}|P_{\text{UE}}$.

% \begin{equation}
% \label{eq:powermodel2}
% \begin{aligned}
% P_{\text{tot}}= & \dfrac{1}{\eta}\sum\limits_{g\in\mathcal{G}}||\mathbf{w}_{g}||_2^{2} +P_{\text{RF}}\sum\limits_{b\in\mathcal{B}}\sum\limits_{i \in \mathcal{N}_b}a_{b,i} + P_0,
% \end{aligned}
% \end{equation}

\vspace{-2mm}
\subsection{Energy Efficiency Maximization}
\textcolor{black}{The appearance of the antenna selection matrices $\mathbf{F}_{b_g}$ in \eqref{eq:SINRexpression} and \eqref{eq:powermodel} make it challenging to proceed further, mostly due to the multiplication $\mathbf{F}_{b_g}\mathbf{w}_{g}$ in both equations. Thus, to find a more tractable formulation, we first remove $\mathbf{F}_{b_g}$ from the expressions and rewrite \eqref{eq:SINRexpression} and \eqref{eq:powermodel}, as
\begin{equation}\label{eq:SINRexpression2}
\Gamma_k(\mathbf{w}) = \frac{|\mathbf{h}_{b_g,k}\herm \mathbf{w}_{g}|^{2}}{{N_0 + \sum\limits_{u \in \mathcal{G} \setminus \{g\}} |\mathbf{h}_{b_u,k}\herm \mathbf{w}_{u}|^{2}}}
\end{equation}
\begin{equation}
\label{eq:powermodel2}
\begin{aligned}
P_{\text{tot}}= & \dfrac{1}{\eta}\sum\limits_{g\in\mathcal{G}}||\mathbf{w}_{g}||_2^{2} +P_{\text{RF}}\sum\limits_{b\in\mathcal{B}}\sum\limits_{i \in \mathcal{N}_b}a_{b,i} + BP_{\text{sta}} + |\mathcal{K}|P_{\text{UE}},
\end{aligned}
\end{equation}
% \begin{equation}
% \label{eq:powermodel2}
% \begin{aligned}
% P_{\text{tot}}= & \dfrac{1}{\eta}\sum\limits_{g\in\mathcal{G}}||\mathbf{w}_{g}||_2^{2} +P_{\text{RF}}\sum\limits_{b\in\mathcal{B}}\sum\limits_{i \in \mathcal{N}_b}a_{b,i} + P_0,
% \end{aligned}
% \end{equation}
where $a_{b,i}\in\{0,1\}$ is the binary antenna selection variable for the $i$th transmit antenna of BS $b$, i.e., $a_{b,i}=1$, if the $i$th antenna is selected, and $a_{b,i}=0$ otherwise. In the antenna selection, we need to ensure that all beamforming coefficients associated with antenna $i$ of BS $b$ should be simultaneously set to zero to switch off the antenna. This connection of the antenna selection variables to the beamforming coefficients is achieved via the constraint $||\hat{\mathbf{w}}_{b,i}||_2^2 \leq a_{b,i}P_{\text{max}}$, where $\hat{\mathbf{w}}_{b,i} \triangleq [\mathbf{w}_{\mathcal{G}_b(1)}[i], \mathbf{w}_{\mathcal{G}_b(2)}[i],\ldots,\mathbf{w}_{\mathcal{G}_b(G_b)}[i]]\trans$ is an expression including the beamforming coefficients related to antenna $i$ of BS $b$. That is, if we set $a_{b,i}=0$, then $||\hat{\mathbf{w}}_{b,i}||_2^2=0$, meaning that in the SINR expression \eqref{eq:SINRexpression2}, $\mathbf{w}_{g}[i]=0, \forall g \in \mathcal{G}_b$. On the other hand, if $a_{b,i}=1$, then this antenna is restricted to have at most the maximum transmit power $P_{\text{max}}$.}

\textcolor{black}{In a multicasting system, the information has to be decodable by all users in a group, which means that the rate for user group $g$ is defined as a minimum of the user rates across the whole group. Thus, we can write the achievable sum rate expression as $R(\mathbf{w})\triangleq \sum\limits_{g\in\mathcal{G}}\underset{k\in\mathcal{K}_g}{\min}\log(1+\Gamma_k(\mathbf{w}))$.}
Under these notations, the network EE maximization problem can be written as
{\begin{subequations} \label{EEmax}
\begin{eqnarray}{} \underset{\mathbf{w},\mathbf{a}}{\maxi} &  & \frac{R(\mathbf{w})}{ g(\mathbf{w},\mathbf{a}) + P_{0}}  \label{eq:EE:obj} \\ \st
 &   & ||\hat{\mathbf{w}}_{b,i}||_2^2 \leq a_{b,i}P_{\text{max}},  \forall b \in \mathcal{B}, i \in \mathcal{N}_b\label{eq:IntPconstraint}\\
% &   & ||\hat{\mathbf{w}}_{b,i}||_2^2 \geq a_{b,i}P_{\text{min}},  \forall b \in \mathcal{B}, i \in \mathcal{N}_b\label{eq:MinIntPconstraint}\\
% * <namletran@gmail.com> 2017-09-27T16:00:03.378Z:
%
% >  &   & ||\hat{\mathbf{w}}_{b,i}||_2^2 \leq a_{b,i}P_{\text{max}},  \forall b \in \mathcal{B}, i \in \mathcal{N}_b\label{eq:IntPconstraint}\\
% You may spend some text on explaining this constraint. We use bigM formulation here  but it normally produces loose continuous relaxation. This motivates the perspective formulation later on.
%
% ^.
 %&   & ||\hat{\mathbf{w}}_{b,i}||_2^2 \geq a_{b,i}P_{\text{min}},  \forall b \in \mathcal{B}, i \in \mathcal{N}_b \label{eq:MinIntPconstraint} \\
 &   & \min_{k\in\mathcal{K}_g}\log(1+\Gamma_k(\mathbf{w})) \geq \max_{k\in\mathcal{K}_g}{\bar{R}}_k,  \forall g \in \mathcal{G}, \label{eq:SINRconstraints}\\
 &   & a_{b,i} \in \{0,1\}, \forall b \in \mathcal{B}, i \in \mathcal{N}_b \label{eq:binary}
\end{eqnarray}
\end{subequations}}
where $g(\mathbf{w},\mathbf{a})\triangleq \sum\limits_{g\in\mathcal{G}}\frac{1}{\eta}||\mathbf{w}_{g}||_2^{2} + P_{\text{RF}}\sum\limits_{b\in\mathcal{B}}\sum\limits_{i\in\mathcal{N}_b}a_{b,i}$ is a function denoting the adjustable power consumption, $\bar{R}_k$ is the minimum rate requirement for user $k$, and $\mathbf{a}\triangleq\{a_{b,i}\}_{b\in\mathcal{B},i\in\mathcal{N}_b}$. Note that for the physical layer multicasting the rate of a certain group is defined by the worst case user. Thus, constraint \eqref{eq:SINRconstraints} is to guarantee that the achieved multicasting rate is larger than the largest QoS requirement in the group, because all the requirements have to be satisfied.
The above problem is a non-convex mixed-Boolean fractional program which is hard to tackle as such. One of the main challenges is that the problem is non-convex even when the Boolean variables are relaxed to be continuous. More specifically, in that case, \eqref{eq:SINRconstraints} and the numerator of the objective function are non-convex.

\vspace{-1mm}
\section{Mixed-Boolean Programming Based Method}\label{sec:CentralizedMethods}

\subsection{Equivalent Transformation}
Here we aim at developing a continuous relaxation based algorithm which yields close to a Boolean solution.
To this end, a tight continuous relaxation plays an important role.
%Thus, we replace \eqref{eq:binary} with $0 \leq a_{b,i} \leq 1$ as a first step. Consequently, \eqref{eq:IntPconstraint} becomes a convex constraint as such after the relaxation.
%As shown later on, this simple  formulation yields a loose continuous  relaxation which makes post-processing difficult and inefficient.
As a first step towards a more efficient reformulation of \eqref{EEmax}, we equivalently replace the maximum power constraints in \eqref{eq:IntPconstraint} with the following two constraints
\begin{subequations}
\label{perspective}
\begin{eqnarray}
&  & \hspace{-10pt} ||\hat{\mathbf{w}}_{b,i}||_2^2 \leq a_{b,i}^\chi v_{b,i},\; \forall b \in \mathcal{B}, i \in \mathcal{N}_b \label{eq:EEmax:reform0:MaxPC}\\
&  & \hspace{-10pt} v_{b,i} \leq P_{\text{max}}, \forall b \in \mathcal{B}, i \in \mathcal{N}_b. \label{eq:EEmax:reform0:vmin}
\end{eqnarray}
\end{subequations}
\textcolor{black}{where the variable  $v_{b,i}$ can be viewed as a soft output power level of antenna $i$ of BS $b$ (i.e., the optimized power when the Boolean variables $a_{b,i}$ are relaxed to continuous), and we have introduced the exponent $\chi\geq 1$ in \eqref{eq:EEmax:reform0:MaxPC} for the sake of a tighter continuous relaxation presented  in details shortly. The equivalence between  \eqref{eq:IntPconstraint} and  \eqref{perspective} is guaranteed as $a_{b,i}$ is Boolean, i.e., $ a_{b,i}^\chi = a_{b,i} $  for any  $\chi>0$.}
\textcolor{black}{Thus, we desire to solve the following equivalent transformation of \eqref{EEmax}
\begin{subequations} \label{EEmaxEqBinary}
\begin{eqnarray}{} \underset{\mathbf{w},\mathbf{a}}{\maxi} &  & \frac{R(\mathbf{w})}{ g(\mathbf{v},\mathbf{a}) + P_{0}}  \label{eq:EE:obj} \\ \st
&    & \eqref{eq:SINRconstraints}, \eqref{eq:binary}, \eqref{eq:EEmax:reform0:MaxPC}, \eqref{eq:EEmax:reform0:vmin}
% &   & ||\hat{\mathbf{w}}_{b,i}||_2^2 \leq a_{b,i}^\chi v_{b,i},\; \forall b \in \mathcal{B}, i \in \mathcal{N}_b\\
% &   & v_{b,i} \leq P_{\text{max}}, \forall b \in \mathcal{B}, i \in \mathcal{N}_b\\
%% &   & ||\hat{\mathbf{w}}_{b,i}||_2^2 \geq a_{b,i}P_{\text{min}},  \forall b \in \mathcal{B}, i \in \mathcal{N}_b\label{eq:MinIntPconstraint}\\
%% * <namletran@gmail.com> 2017-09-27T16:00:03.378Z:
%%
%% >  &   & ||\hat{\mathbf{w}}_{b,i}||_2^2 \leq a_{b,i}P_{\text{max}},  \forall b \in \mathcal{B}, i \in \mathcal{N}_b\label{eq:IntPconstraint}\\
%% You may spend some text on explaining this constraint. We use bigM formulation here  but it normally produces loose continuous relaxation. This motivates the perspective formulation later on.
%%
%% ^.
% %&   & ||\hat{\mathbf{w}}_{b,i}||_2^2 \geq a_{b,i}P_{\text{min}},  \forall b \in \mathcal{B}, i \in \mathcal{N}_b \label{eq:MinIntPconstraint} \\
% &   & \log(1+\Gamma_k(\mathbf{w})) \geq {\bar{R}}_k,  \forall k \in \mathcal{K}, \label{eq:SINRconstraints}\\
% &   & a_{b,i} \in \{0,1\}, \forall b \in \mathcal{B}, i \in \mathcal{N}_b
\end{eqnarray}
\end{subequations}
where $\mathbf{v}\triangleq\{v_{b,i}\}_{b\in\mathcal{B},i\in\mathcal{N}_b}$ and $g(\mathbf{v},\mathbf{a})\triangleq \sum\limits_{b\in\mathcal{B}}\sum\limits_{i\in\mathcal{N}_b}\frac{1}{\eta}v_{b,i} + P_{\text{RF}}\sum\limits_{b\in\mathcal{B}}\sum\limits_{i\in\mathcal{N}_b}a_{b,i}$.}
%The equivalence between  \eqref{eq:IntPconstraint} and  \eqref{perspective} is guaranteed as $a_{b,i}$ is Boolean, i.e., $ a_{b,i}^\chi = a_{b,i} $  for any  $\chi>0$. We note that when  $\chi= 1$, \eqref{perspective} is called the perspective formulation \cite{Gunluk-10,MINLP} which is usually used to efficiently find  optimal solutions for mixed-Boolean programs with convex continuous relaxations.  However, in our case the considered problem is non-convex even after the continuous relaxation, that is, we aim at finding an approximate solution for the continuous relaxation which can be easily mapped to a high-quality solution for the original mixed-Boolean fractional program.
%Specifically, the problem with the original formulation \eqref{eq:MinIntPconstraint} is that solving the continuous relaxation of \eqref{EEmax} results in non-zero continuous values of $a_{b,i}$.
% Thus, it is very difficult to map the continuous values to binary.
We  remark that it is natural to write $g(\mathbf{v},\mathbf{a})= \sum\limits_{b\in\mathcal{B}}\sum\limits_{i\in\mathcal{N}_b}\frac{1}{\eta}v_{b,i}a_{b,i} + P_{\text{RF}}\sum\limits_{b\in\mathcal{B}}\sum\limits_{i\in\mathcal{N}_b}a_{b,i}$. To achieve a more tractable formulation, we define $g(\mathbf{v},\mathbf{a})$ as done in \eqref{EEmaxEqBinary}, i.e., $a_{b,i}$ is excluded from the first term. However, \eqref{EEmax} and \eqref{EEmaxEqBinary} are still equivalent in the sense that they achieve the same optimal solutions, which can be proved as follows. Firstly we note that $a_{b,i}$ is binary in both problems. Secondly, (8a) has to be satisfied with equality at the optimality. Otherwise we can strictly decrease  $v_{b,i}$ without violating (8a) but then achieve a larger objective value for \eqref{EEmaxEqBinary}. Now it is clear that if $a_{b,i}=0$, then both $v_{b,i}=0$ and $||\hat{\mathbf{w}}_{b,i}||_2^2=0$. Also if $a_{b,i}=1$, then $v_{b,i}=||\hat{\mathbf{w}}_{b,i}||_2^2$ as $a_{b,i}=a_{b,i}^{\chi}=1$.
%  However, our numerical experiments show the effectiveness of the proposed approach. On the other hand, \eqref{eq:EEmax:reform1:vMinPC} could be ignored in \eqref{eq:EEmax:reform1}, because the constraints \eqref{eq:EEmax:reform1:MaxPC}-\eqref{eq:EEmax:reform1:MinPC} already guarantee \eqref{eq:EEmax:reform1:vMinPC}. However, when solving the continuous relaxation of \eqref{eq:EEmax:reform1}, \eqref{eq:EEmax:reform1:vMinPC} limits $v_{b,i}$'s so that they cannot become too small.

The motivation for introducing the exponent $\chi\geq 1$ in \eqref{eq:EEmax:reform0:MaxPC} is explained as follows. First, we note that when  $\chi= 1$, \eqref{eq:EEmax:reform0:MaxPC} is called the perspective formulation \cite{Gunluk-10,MINLP}, and both \eqref{eq:EEmax:reform0:MaxPC}  and \eqref{eq:EEmax:reform0:vmin} are convex. Thus, the perspective formulation is routinely used to  find  optimal solutions for mixed-Boolean programs with convex continuous relaxations e.g. in \cite{Cheng-13}.  However, this is not  case for the continuous relaxation of the considered problem in \eqref{EEmax} due to \eqref{eq:SINRconstraints} and the numerator of the objective function. As later on we adopt the SCA to find a suboptimal solution to the continuous relaxation, a tight continuous relaxation of \eqref{EEmax} is critically important as it increases the chance of obtaining a high-quality solution  for the
original mixed-Boolean fractional program. Although this cannot be analytically proved, it is intuitively explained as follows. \textcolor{black}{The role of exponent $\chi\geq 1$ in \eqref{eq:EEmax:reform0:MaxPC} is to act as a penalty parameter which penalizes the values of $a_{b,i}$ so that they are encouraged towards a Boolean solution when considering the continuous relaxation. More explicitly, the larger $\chi$, the tighter is the continuous relaxation. }\textcolor{black}{Mathematically, we have the following.
\begin{lm}\label{Lemma1} Let $\text{EE}_{\text{bool}}$, $\text{EE}_{\text{cont,}\chi=m}$, and $\text{EE}_{\text{cont,orig}}$ refer to the optimal objective of the Boolean formulation \eqref{EEmax}, continuous relaxation of \eqref{EEmaxEqBinary} with $\chi=m$, and continuous relaxation of \eqref{EEmax}. Then the following inequality holds
\begin{equation}
\text{EE}_{\text{bool}}  \overset{(iii)}{\leq}  \text{EE}_{\text{cont,$\chi=m$}} \overset{(ii)}{\leq} \text{EE}_{\text{cont,$\chi=1$}} \overset{(i)}{\leq} \text{EE}_{\text{cont,orig}}
\end{equation}
\end{lm}
\begin{proof}
See Appendix \ref{App1}.
\end{proof}}
% In fact, the proposed formulation is a generalization of the approach used in \cite{Tervo-15} for single-cell multiuser MISO unicasting, where $\chi=1$ and $\chi=2$ were considered.
The above lemma states that the optimal objective of the proposed continuous relaxation becomes closer to that of the original mixed-Boolean program as $\chi$ increases. Thus, it is reasonable to expect that solving the continuous relaxation with a proper choice of $\chi$  and rounding the obtained solution may provide a good solution for the original problem. This is numerically verified in Section \ref{sec:NumericalResults}.
\vspace{-2mm}
\textcolor{black}{\subsection{Proposed Method to Solve \eqref{EEmaxEqBinary}}
We propose an algorithm which aims to find a good solution to \eqref{EEmax} (or, equivalently \eqref{EEmaxEqBinary}). The algorithm consists of two phases:} 1) \textcolor{black}{solving continuous relaxation of \eqref{EEmaxEqBinary} and} 2) \textcolor{black}{recovering the Boolean solution from the continuous relaxation.}

\subsubsection{Solving Continuous Relaxation of \eqref{EEmaxEqBinary}}\label{SolveContRelax}

\textcolor{black}{The problem of interest can be written as
\begin{subequations} \label{EEmax2}
\begin{eqnarray}{} \underset{\mathbf{w},\mathbf{v},\mathbf{a}}{\maxi} &  & \frac{\sum_{g\in\mathcal{G}}\min_{k\in\mathcal{K}_g}\log(1+\Gamma_k(\mathbf{w}))}{ g(\mathbf{v},\mathbf{a}) + P_{0}}  \label{eq:EE4:obj} \\ \st
%&  & v_{b,i} \leq P_{\text{max}}, \forall b \in \mathcal{B}, i \in \mathcal{N}_b. \label{eq:EEmax:reform0:vmin}\\
% &   & ||\hat{\mathbf{w}}_{b,i}||_2^2 \geq a_{b,i}P_{\text{min}},  \forall b \in \mathcal{B}, i \in \mathcal{N}_b\label{eq:MinIntPconstraint}\\
% * <namletran@gmail.com> 2017-09-27T16:00:03.378Z:
%
% >  &   & ||\hat{\mathbf{w}}_{b,i}||_2^2 \leq a_{b,i}P_{\text{max}},  \forall b \in \mathcal{B}, i \in \mathcal{N}_b\label{eq:IntPconstraint}\\
% You may spend some text on explaining this constraint. We use bigM formulation here  but it normally produces loose continuous relaxation. This motivates the perspective formulation later on.
%
% ^.
 %&   & ||\hat{\mathbf{w}}_{b,i}||_2^2 \geq a_{b,i}P_{\text{min}},  \forall b \in \mathcal{B}, i \in \mathcal{N}_b \label{eq:MinIntPconstraint} \\
 &   & \min_{k\in\mathcal{K}_g}\log(1+\Gamma_k(\mathbf{w})) \geq \max_{k\in\mathcal{K}_g}{\bar{R}}_k,  \forall g \in \mathcal{G}, \label{eq:SINRconstraints2}\\
 &   & 0 \leq a_{b,i} \leq 1, \forall b\in\mathcal{B}, i\in\mathcal{N}_b \label{eq:EEmax:reform1:relaxed} \\
% &   & a_{b,i} \in \{0,1\}, \forall b \in \mathcal{B}, i \in \mathcal{N}_b \label{eq:binary} \\
 &   & \eqref{eq:EEmax:reform0:MaxPC}, \eqref{eq:EEmax:reform0:vmin}
\end{eqnarray}
\end{subequations}
%where the rate of group $g$ is defined as $\min_{k\in\mathcal{K}_g}\log(1+\Gamma_k(\mathbf{w}))$.
%\subsubsection{Solving the Continuous Relaxation \eqref{EEmax2}}
Then, we replace $\min_{k\in\mathcal{K}_g}\log(1+\Gamma_k(\mathbf{w}))$ with a new variable $r_g=\min_{k\in\mathcal{K}_g}\log(1+\Gamma_k(\mathbf{w}))$ and rewrite the above problem equivalently as
\begin{subequations} \label{EEmax3}
\begin{eqnarray}{} \underset{\mathbf{w},\mathbf{v},\mathbf{a},\mathbf{r}}{\maxi} &  & \frac{\sum_{g\in\mathcal{G}}r_g}{ g(\mathbf{v},\mathbf{a}) + P_{0}}  \label{eq:EE5:obj} \\ \st
%&  &  ||\hat{\mathbf{w}}_{b,i}||_2^2 \leq a_{b,i}^\chi v_{b,i},\; \forall b \in \mathcal{B}, i \in \mathcal{N}_b \label{eq:EEmax:reform0:MaxPC}\\
%&  & v_{b,i} \leq P_{\text{max}}, \forall b \in \mathcal{B}, i \in \mathcal{N}_b. \label{eq:EEmax:reform0:vmin}\\
% &   & ||\hat{\mathbf{w}}_{b,i}||_2^2 \geq a_{b,i}P_{\text{min}},  \forall b \in \mathcal{B}, i \in \mathcal{N}_b\label{eq:MinIntPconstraint}\\
% * <namletran@gmail.com> 2017-09-27T16:00:03.378Z:
%
% >  &   & ||\hat{\mathbf{w}}_{b,i}||_2^2 \leq a_{b,i}P_{\text{max}},  \forall b \in \mathcal{B}, i \in \mathcal{N}_b\label{eq:IntPconstraint}\\
% You may spend some text on explaining this constraint. We use bigM formulation here  but it normally produces loose continuous relaxation. This motivates the perspective formulation later on.
%
% ^.
 %&   & ||\hat{\mathbf{w}}_{b,i}||_2^2 \geq a_{b,i}P_{\text{min}},  \forall b \in \mathcal{B}, i \in \mathcal{N}_b \label{eq:MinIntPconstraint} \\
  &   & r_g = \min_{k\in\mathcal{K}_g}\log(1+\Gamma_k(\mathbf{w})),  \forall g \in \mathcal{G}, \label{eq:equalityconst}\\
 &   & r_g \geq \max_{k\in\mathcal{K}_g}{\bar{R}}_k,  \forall g \in \mathcal{G}, \label{eq:EEmax:reform1:minSINR}\\
  &   & \eqref{eq:EEmax:reform0:MaxPC}, \eqref{eq:EEmax:reform0:vmin}, \eqref{eq:EEmax:reform1:relaxed}
 %&   & a_{b,i} \in \{0,1\}, \forall b \in \mathcal{B}, i \in \mathcal{N}_b \label{eq:binary}
\end{eqnarray}
\end{subequations}
where $\mathbf{r}\triangleq \{r_g\}_{g\in\mathcal{G}}$.
In the above, \eqref{eq:equalityconst} can be further replaced by the inequality $r_g \leq \min_{k\in\mathcal{K}_g}\log(1+\Gamma_k(\mathbf{w}))$, which is then equivalent to $r_g \leq \log(1+\Gamma_k(\mathbf{w})), \forall k \in \mathcal{K}_g$. To address the nonconvex rate function, we introduce new variables $\boldsymbol\gamma\triangleq\{\gamma_k\}_{k\in\mathcal{K}}$ to represent the SINR of each user $k$ \cite{Oskari:2017:TSP}, and write \eqref{EEmax3} equivalently as
\begin{subequations}
\label{eq:EEmax:reform1}
%\begin{align}
\begin{eqnarray}
\hspace{-10pt} \underset{\mathbf{w},\boldsymbol\gamma, \mathbf{v},\mathbf{a},\mathbf{r}}{\maxi} &  &  \frac{\sum_{g\in\mathcal{G}}r_g}{ g(\mathbf{v},\mathbf{a}) + P_{0}}\label{EqObj}\\
\st
%&  & \hspace{-10pt} ||\hat{\mathbf{w}}_{b,i}||_2^2 \leq a_{b,i}^\chi v_{b,i},\; \forall b \in \mathcal{B}, i \in \mathcal{N}_b \label{eq:EEmax:reform1:MaxPC}\\
%&  & \hspace{-10pt} v_{b,i} \leq P_{\text{max}}, \forall b \in \mathcal{B}, i \in \mathcal{N}_b \label{eq:EEmax:reform1:vmin}\\
%&   & \hspace{-10pt} ||\hat{\mathbf{w}}_{b,i}||_2^2 \geq a_{b,i}P_{\text{min}},  \forall b \in \mathcal{B}, i \in \mathcal{N}_b\label{eq:reform1:MinIntPconstraint}\\
%&  & \hspace{-10pt} v_{b,i} \geq a_{b,i} P_{\text{min}}, \forall b \in \mathcal{B}, i \in \mathcal{N}_b \label{eq:EEmax:reform1:vMinPC} \\
%&  & \hspace{-10pt} ||\hat{\mathbf{w}}_{b,i}||_2^2 \geq a_{b,i}^\alpha P_{\text{min}}, \forall b \in \mathcal{B}, i \in \mathcal{N}_b \label{eq:EEmax:reform1:MinPC} \\
%&  & \hspace{-10pt} v_{b,i} \geq a_{b,i} P_{\text{min}}, \forall b \in \mathcal{B}, i \in \mathcal{N}_b \label{eq:EEmax:reform1:vMinPC} \\
%&  & \hspace{-10pt} 0 \leq a_{b,i} \leq 1, \forall b \in \mathcal{B}, i \in \mathcal{N}_b \label{eq:EEmax:reform1:relaxed}\\
&  & \hspace{-10pt} \gamma_k \leq \frac{|\mathbf{h}_{b_g,k}\herm \mathbf{w}_{g}|^{2}}{{N_0 + \sum\limits_{u \in \mathcal{G} \setminus \{g\}} |\mathbf{h}_{b_u,k}\herm \mathbf{w}_{u}|^{2}}}, \forall k \in \mathcal{K}\label{eq:EEmax:reform1:rate}\\
%&  & \hspace{-10pt} r_g \geq \max_{k\in\mathcal{K}_g}(\bar{R}_k), \forall g \in \mathcal{G}\label{eq:EEmax:reform1:minSINR}\\
&  & \hspace{-10pt} r_g \leq \log(1+\gamma_k), \forall g \in \mathcal{G}, k \in \mathcal{K}_g\label{eq:EEmax:reform1:weakestRATE}\\
&  & \hspace{-10pt} \eqref{eq:EEmax:reform1:relaxed}, \eqref{eq:EEmax:reform1:minSINR}, \eqref{eq:EEmax:reform0:MaxPC}, \eqref{eq:EEmax:reform0:vmin}
\end{eqnarray}
\end{subequations}  }
\textcolor{black}{\begin{lm}\label{Lemma2}
\label{lem:equivalence}Problems \eqref{EEmax2} and \eqref{eq:EEmax:reform1}
are equivalent at the optimality.
\end{lm}
\begin{proof}
See Appendix \ref{App2}.
\end{proof}}

%The equivalence of \eqref{EEmax3} and \eqref{eq:EEmax:reform1} is not difficult to observe, because the constraints in \eqref{eq:EEmax:reform1:rate} are active and $r_g = \underset{k\in\mathcal{K}_g}{\min}\log(1+\gamma_k), \forall g \in \mathcal{G}$ in \eqref{eq:EEmax:reform1:weakestRATE} is satisfied at the optimum.
%The constraint $r_g\leq\underset{k\in\mathcal{K}_g}{\min}\log(1+\Gamma_k(\mathbf{w}))$ has been equivalently replaced with $r_g\leq\log(1+\Gamma_k(\mathbf{w})), \forall k \in \mathcal{K}_g$ in \eqref{eq:EEmax:reform1:weakestRATE}.
%This is because if the minimum user rate satisfies the constraint, then all the user rates in that group have to satisfy it. On the other hand, the  minimum group rate for group $g$ in \eqref{eq:EEmax:reform1:minSINR} has to be at least as large as the largest rate requirement in the group.
\vspace{-2mm}
By looking at the formulation \eqref{eq:EEmax:reform1}, it is discovered that the objective function is a concave-convex fractional function and the main challenge in solving \eqref{eq:EEmax:reform1} is in the constraints \eqref{eq:EEmax:reform1:rate} and \eqref{eq:EEmax:reform0:MaxPC}.
To handle these, we use the same idea as that in \cite{Oskari:2017:TSP,Tervo-15,Venkatraman-16} to replace \eqref{eq:EEmax:reform1:rate} equivalently as
\begin{subequations}
\label{eq:rep:reform2}
%\begin{align}
\begin{eqnarray}
&  & \hspace{-15pt} \gamma_k \leq \frac{|\mathbf{h}_{b_g,k}\herm \mathbf{w}_{g}|^{2}}{\beta_k}, \forall k \in \mathcal{K} \label{eq:rep:reform2:quadoverlin}\\
%&  & \gamma_k \geq \bar{\Gamma}_k \\
&  & \hspace{-15pt} \beta_k \geq {N_0 + \sum\limits_{u \in \mathcal{G} \setminus \{g\}} |\mathbf{h}_{b_u,k}\herm \mathbf{w}_{u}|^{2}}, \forall k \in \mathcal{K}\label{eq:rep:reform2:betaconst}
\end{eqnarray}
\end{subequations}
where $\boldsymbol\beta \triangleq \{\beta_k\}_{k\in\mathcal{K}}$ are new variables representing the total interference-plus-noise of user $k$.
Now \eqref{eq:rep:reform2:betaconst} is readily a convex constraint, while \eqref{eq:rep:reform2:quadoverlin} involves a convex function at both sides. Specifically, the left and right sides of \eqref{eq:rep:reform2:quadoverlin} are linear and quadratic-over-linear functions, respectively.  To formulate \eqref{eq:EEmax:reform0:MaxPC} in a more tractable manner, we first write the following equivalent form
\begin{equation}\label{quadlinpers}
\frac{||\hat{\mathbf{w}}_{b,i}||_2^2}{v_{b,i}} \leq a_{b,i}^\chi ,\; \forall b \in \mathcal{B}, i \in \mathcal{N}_b.
\end{equation}
In \eqref{quadlinpers}, \textcolor{black}{the left side is a convex quadratic-over-linear function, and the right side is also convex \cite{Boyd:ConvexOpt:2004}}. At this point, we can equivalently write \eqref{eq:EEmax:reform1} as
\begin{subequations}
\label{eq:EEmax:reform2}
%\begin{align}
\begin{eqnarray}
\hspace{-15pt} \underset{\mathbf{w},\boldsymbol\gamma, \mathbf{v},\mathbf{a},\boldsymbol\beta,\mathbf{r}}{\maxi} &  &  \frac{\sum_{g\in\mathcal{G}}r_g}{ g(\mathbf{v},\mathbf{a}) + P_{0}}\\
\st
&  &  \hspace{-15pt} \frac{||\hat{\mathbf{w}}_{b,i}||_2^2}{v_{b,i}} \leq a_{b,i}^\chi ,\; \forall b \in \mathcal{B}, i \in \mathcal{N}_b \label{eq:EEmax:reform2:MaxPC}\\
%&  & v_{b,i} \leq P_{\text{max}}\\
%&  & \hspace{-15pt} ||\hat{\mathbf{w}}_{b,i}||_2^2 \geq a_{b,i}^\alpha P_{\text{min}}, \forall b \in \mathcal{B}, i \in \mathcal{N}_b \label{eq:EEmax:reform2:MinPC} \\
%&   & \hspace{-15pt} ||\hat{\mathbf{w}}_{b,i}||_2^2 \geq a_{b,i}P_{\text{min}},  \forall b \in \mathcal{B}, i \in \mathcal{N}_b\label{eq:reform2:MinIntPconstraint}\\
&  & \hspace{-15pt} \gamma_k \leq \frac{|\mathbf{h}_{b_g,k}\herm \mathbf{w}_{g}|^{2}}{\beta_k}, \forall k \in \mathcal{K} \label{eq:EEmax:reform2:quadoverlin}\\
&  & \hspace{-15pt} 0 \leq a_{b,i} \leq 1, \forall b \in \mathcal{B}, i \in \mathcal{N}_b\\
%&  & \gamma_k \geq \bar{\Gamma}_k \\
%&  & \hspace{-15pt} \beta_k \geq {\sigma_{k}^{2} + \sum\limits_{u \in \mathcal{G} \setminus \{g\}} |\mathbf{h}_{b_u,k}\herm \mathbf{w}_{u}|^{2}}, \forall k \in \mathcal{K}\label{eq:EEmax:reform2:betaconst}\\
&  & \hspace{-15pt} \eqref{eq:EEmax:reform0:vmin}, \eqref{eq:EEmax:reform1:minSINR}, \eqref{eq:EEmax:reform1:weakestRATE}, \eqref{eq:rep:reform2:betaconst}.
\end{eqnarray}
\end{subequations}
Now we can see that in \eqref{eq:EEmax:reform2}, all the other constraints are convex except \eqref{eq:EEmax:reform2:MaxPC}, and \eqref{eq:EEmax:reform2:quadoverlin}, which can be expressed as a difference of convex functions. We propose to use successive convex approximation to approximate \eqref{eq:EEmax:reform2} as a convex problem in each iteration. Specifically, \textcolor{black}{at some iteration $n$ of the SCA, the nonconvex parts of \eqref{eq:EEmax:reform2:quadoverlin} and \eqref{eq:EEmax:reform2:MaxPC} are approximated by convex ones at some operating point with the aid of the first-order Taylor approximations.}
To deal with the right side of \eqref{eq:EEmax:reform2:quadoverlin}, we can write its linear \textcolor{black}{first-order Taylor lower bound approximation} at point $(\mathbf{w}_{g}^{(n)},\beta_{k}^{(n)})$ as
\begin{eqnarray}\label{eq:EEmax:quad_over_lin_app}
 |\mathbf{h}_{b_g,k}\herm\mathbf{w}_{g}|^2/\beta_{k} \geq 2\Re((\mathbf{w}_{g}^{(n)})\herm\mathbf{h}_{b_g,k}\mathbf{h}_{b_g,k}\herm\mathbf{w}_{g})/\beta_{k}^{(n)}\nonumber \\
 - (|\mathbf{h}_{b_g,k}\herm\mathbf{w}_{g}^{(n)}|/\beta_{k}^{(n)})^2\beta_{k} \triangleq \Psi_k^{(n)}(\mathbf{w}_{g},\beta_{k}).
\end{eqnarray}
 For \eqref{eq:EEmax:reform2:MaxPC}, we can write the linear lower bound approximation of the right side at point $a_{b,i}^{(n)}$ as
%\begin{eqnarray}\label{eq:EEmax:quad_over_lin_app}
% |\mathbf{h}_{b_k,k}\mathbf{w}_{g}|^2/\beta_{k} \geq 2\Re((\mathbf{w}_{g}^{(n)})\herm\mathbf{h}_{b_k,k}\herm\mathbf{h}_{b_k,k}\mathbf{w}_{g})/\beta_{k}^{(n)}\nonumber \\
% - (|\mathbf{h}_{b_k,k}\mathbf{w}_{g}^{(n)}|/\beta_{k}^{(n)})^2\beta_{k} \triangleq \Psi_k^{(n)}(\mathbf{w}_{g},\beta_{k}).
%\end{eqnarray}
\begin{eqnarray}\label{eq:EEmax:Approx}
 a_{b,i}^\chi \geq (1-\chi)(a_{b,i}^{(n)})^\chi + \chi {(a_{b,i}^{(n)})}^{(\chi-1)}a_{b,i} \triangleq \Upsilon_{b,i}^{(n)}(a_{b,i}).
\end{eqnarray}
%Finally, the linear lower approximation for the left side of \eqref{eq:reform1:MinIntPconstraint} at point $\hat{\mathbf{w}}_{b,i}^{(n)}$ can be stated as
%\begin{eqnarray}\label{eq:EEmax:minPCapprox}
% ||\hat{\mathbf{w}}_{b,i}||_2^2 \geq 2\Re((\hat{\mathbf{w}}_{b,i}^{(n)})\herm\hat{\mathbf{w}}_{b,i}) - ||\hat{\mathbf{w}}_{b,i}^{(n)}||_2^2  \triangleq \Phi_{b,i}^{(n)}(\hat{\mathbf{w}}_{b,i}).
%\end{eqnarray}
With the approximations \eqref{eq:EEmax:quad_over_lin_app} and \eqref{eq:EEmax:Approx} we can write the concave-convex fractional problem {\color{black}at iteration $n+1$} \textcolor{black}{of the SCA} as
\begin{subequations}
\label{eq:EEmax:reform3}
%\begin{align}
\begin{eqnarray}
\hspace{-20pt} \underset{\mathbf{w},\boldsymbol\gamma,\mathbf{v},\mathbf{a},\boldsymbol\beta,\mathbf{r}}{\maxi} &  &  \frac{\sum_{g\in\mathcal{G}}r_g}{ g(\mathbf{v},\mathbf{a}) + P_{0}}\label{eq:EEmax:reform3:obj}\\
\st
%&  &  ||\hat{\mathbf{w}}_{b,i}||_2^2 \leq a_{b,i}v_{b,i},\; \forall b \in \mathcal{B}, i \in \mathcal{N}_b \label{eq:EE:PC}\\
%&  & v_{b,i} \leq P_{\text{max}}\\
%&  & \Phi_{b,i}^{(n)}(\hat{\mathbf{w}}_{b,i}) \geq a_{b,i}^\alpha P_{\text{min}}, \forall b \in \mathcal{B}, i \in \mathcal{N}_b \label{eq:EEmax:reform3:MinPC} \\
&  & \frac{||\hat{\mathbf{w}}_{b,i}||_2^2}{v_{b,i}} \leq \Upsilon_{b,i}^{(n)}(a_{b,i}), \forall b \in \mathcal{B}, i \in \mathcal{N}_b \label{eq:EEmax:reform3:ApproxB}\\
%&  & \Phi_{b,i}^{(n)}(\hat{\mathbf{w}}_{b,i}) \geq a_{b,i} P_{\text{min}}, \forall b \in \mathcal{B}, i \in \mathcal{N}_b \label{eq:EEmax:reform3:MinPC} \\
&  & \gamma_k \leq \Psi_k^{(n)}(\mathbf{w}_{g},\beta_{k}), \forall k \in \mathcal{K}\label{eq:EEmax:reform3:approxquad}\\
&  & 0 \leq a_{b,i} \leq 1, \forall b \in \mathcal{B}, i \in \mathcal{N}_b\label{eq:EEmax:reform3:cont}\\
%&  & \gamma_k \geq \bar{\Gamma}_k \\
%&  & \beta_k \geq {\sigma_{u}^{2} + \sum\limits_{j \in \mathcal{B}} \sum\limits_{k \in \mathcal{G}_j \setminus \{g\}} |\mathbf{h}_{j,u}\herm \mathbf{w}_{k}|^{2}} \label{eq:EEmax:reform3:betaconst}
&  & \eqref{eq:EEmax:reform0:vmin}, \eqref{eq:EEmax:reform1:minSINR},\eqref{eq:rep:reform2:betaconst}, \eqref{eq:EEmax:reform1:weakestRATE}.
\end{eqnarray}
\end{subequations}
\textcolor{black}{Note that although the objective of \eqref{eq:EEmax:reform3} is a linear-fractional function, \eqref{eq:EEmax:reform3} is not classified as a linear-fractional  program as its convex constraints are not linear.}
{\color{black} We also note that \eqref{eq:EEmax:reform3} is not convex but its optimal solution can be found efficiently. This problem is further discussed in the following paragraph.} In the proposed algorithm, the successive convex approximation \cite{Beck:SCA:2010} framework is used, where the concave-convex fractional problem \eqref{eq:EEmax:reform3} is solved at iteration $n+1$. After solving the problem at iteration $n+1$, the optimal solutions $\mathbf{w}_g^*, \beta_k^*, a_{b,i}^*$ are then used to update $\Psi_k^{(n+1)}(\mathbf{w}_{g},\beta_{k})$ and $\Upsilon_{b,i}^{(n+1)}(a_{b,i})$ for the next iteration. The monotonic convergence of the objective function \eqref{eq:EEmax:reform3:obj} is not difficult to see, and a detailed convergence analysis for the problem with similar structure can be found, e.g., in \cite[Appendix A]{Venkatraman-16}. To be self-contained, a convergence proof of the proposed iterative algorithm is provided in Appendix \ref{app:convergence}.

%\vspace{-1.2mm}
We now present efficient ways to solve \eqref{eq:EEmax:reform3}. As mentioned above, \eqref{eq:EEmax:reform3} is a concave-convex fractional program for which two common methods can be used to find an optimal solution: the Dinkelbach's method or the Charnes-Cooper transformation \cite{Schaible-76}. We simply adopt the latter which transforms \eqref{eq:EEmax:reform3} into the following equivalent convex form: \begin{subequations}
\label{eq:EEmax:reform4}
%\begin{align}
\begin{eqnarray}
& \hspace{-20pt} \underset{\phi>0,\bar{\mathbf{w}},\bar{\boldsymbol\gamma},\bar{\mathbf{v}},\bar{\mathbf{a}},\bar{\boldsymbol\beta},\bar{\mathbf{r}}}{\maxi} &   \sum_{g\in\mathcal{G}}\bar{r}_g\label{eq:EEmax:reform4:obj}\\
& \hspace{-60pt} \st  & \hspace{-35pt} \sum\limits_{b\in\mathcal{B}}\sum\limits_{i\in\mathcal{N}_b}(\frac{1}{\eta}\bar{v}_{b,i} + P_{\text{RF}}\bar{a}_{b,i}) + \phi P_{0} \leq 1  \\
%&  & \hspace{-25pt}  ||\bar{\hat{\mathbf{w}}}_{b,i}||_2^2 \leq \bar{a}_{b,i}\bar{v}_{b,i},\; \forall b \in \mathcal{B}, i \in \mathcal{N}_b \label{eq:EE:PC}\\
&  & \hspace{-35pt} \frac{||\bar{\hat{\mathbf{w}}}_{b,i}||_2^2}{\bar{v}_{b,i}} \leq \phi\Upsilon_{b,i}^{(n)}(\frac{\bar{a}_{b,i}}{\phi}), \forall b \in \mathcal{B}, i \in \mathcal{N}_b\\
&  & \hspace{-35pt} \bar{v}_{b,i} \leq \phi P_{\text{max}}, \forall b \in \mathcal{B}, i \in \mathcal{N}_b \\
%&  & \hspace{-25pt} \Phi_{b,i}^{(n)}(\bar{\hat{\mathbf{w}}}_{b,i},\phi) \geq \phi(\frac{\bar{a}_{b,i}}{\phi})^\alpha P_{\text{min}}, \forall b \in \mathcal{B}, i \in \mathcal{N}_b \label{eq:EEmax:reform2:MinPC} \\
%&  & \hspace{-25pt} \bar{v}_{b,i} \geq \bar{a}_{b,i} P_{\text{min}}, \forall b \in \mathcal{B}, i \in \mathcal{N}_b \label{eq:EEmax:reform2:vMinPC} \\
&  & \hspace{-35pt} 0 \leq \bar{a}_{b,i} \leq \phi, \forall b \in \mathcal{B}, i \in \mathcal{N}_b\\
&  & \hspace{-35pt} \bar{\gamma}_k \leq \phi\Psi_k^{(n)}(\frac{\bar{\mathbf{w}}_{g}}{\phi},\frac{\bar{\beta}_{k}}{\phi}), \forall k \in \mathcal{K} \label{eq:EEmax:reform4:SINRapprox} \\
%&  & \hspace{-35pt} \phi \Phi_{b,i}^{(n)}(\frac{\hat{\mathbf{w}}_{b,i}}{\phi}) \geq (\bar{a}_{b,i} -\bar{\rho}_{b,i}) P_{\text{min}}, \forall b \in \mathcal{B}, i \in \mathcal{N}_b \label{eq:EEmax:reform4:MinPC} \\
&  & \hspace{-35pt} \bar{r}_g \geq \phi \max_{k\in\mathcal{K}_g}(\bar{R}_k), \forall g \in \mathcal{G} \\
&  & \hspace{-35pt} \phi\bar{\beta}_k \geq {\phi^2 N_0 + \sum\limits_{u \in \mathcal{G} \setminus \{g\}} |\mathbf{h}_{b_u,k}\herm \bar{\mathbf{w}}_{u}|^{2}}, \forall k \in \mathcal{K}\\
&  & \hspace{-35pt} \bar{r}_g \leq \phi\log(1+\frac{\bar{\gamma}_k}{\phi}), \forall g \in \mathcal{G}, k \in \mathcal{K}_g \label{eq:EEmax:reform4:rateconst}.
\end{eqnarray}
\end{subequations}
From the solution of \eqref{eq:EEmax:reform4}, the optimal solution for the original fractional program \eqref{eq:EEmax:reform3} can be extracted as $\mathbf{w}_g^* = \bar{\mathbf{w}}_g^*/\phi^*, \beta_k^* = \bar{\beta}_k^*/\phi^*,\gamma_k^* = \bar{\gamma}_k^*/\phi^*,a_{b,i}^* = \bar{a}_{b,i}^*/\phi^*,v_{b,i}^* = \bar{v}_{b,i}^*/\phi^*,r_{g}^* = \bar{r}_{g}^*/\phi^*$, where $\bar{\mathbf{w}}_g^*,\phi^*, \bar{\beta}_k^*,\bar{\gamma}_k^*,\bar{a}_{b,i}^*,\bar{v}_{b,i}^*,\bar{r}_{g}^*$ are the optimal variables of \eqref{eq:EEmax:reform4}. The variable $\phi$ represents the inverse of the total power consumption in the problem.
% In the algorithm, the successive convex approximation \cite{Beck:SCA:2010} framework is used, where the convex problem \eqref{eq:EEmax:reform4} is solved at iteration $n$. After solving the problem at iteration $n$, the optimal solutions $\mathbf{w}_g^*, \beta_k^*, a_{b,i}^*$ are then used to update $\Psi_k^{(n+1)}(\mathbf{w}_{g},\beta_{k})$ and $\chi_{b,i}^{(n+1)}(a_{b,i})$ for the next iteration.
% The monotonic convergence of the objective function \eqref{eq:EEmax:reform4:obj} is not difficult to see, and a detailed convergence analysis for the problem with similar structure can be found, e.g., in \cite[Appendix A]{Venkatraman-16}.

\begin{rem}\label{rem:minAnt}
If at least one of the rate targets $\bar{R}_k$ in some user group of BS $b$ is non-zero, we can further reduce the feasible set of \eqref{eq:EEmax:reform4} by adding the constraints
\begin{equation}
\sum_{i\in\mathcal{N}_b}\bar{a}_{b,i} \geq \phi X_b, \forall b \in \mathcal{B}
\end{equation}
where $X_b$ is the number of groups served by BS $b$ which have at least one user having non-zero rate target. This can be done because it is known that at least $X_b$ antennas have to be active to be able to transmit $X_b$ independent data streams.
\end{rem}

\subsubsection{Recovering the Boolean Solution from Continuous Relaxation}
Generally, solving the continuous relaxation usually results in a solution where many of the antenna selection variables are non-Boolean.
% This is even more general when the minimum power constraints are involved.
 However, due to the new formulation in \eqref{eq:EEmax:reform2:MaxPC}, many of the continuous antenna selection variables converge \textcolor{black}{either close to zero (i.e., $a_{b,i}<\epsilon$), or close to 1 (i.e., $a_{b,i}>1-\epsilon$), where $\epsilon$ is a small threshold. Accordingly, those small $a_{b,i}$'s can be directly set to 0 and $a_{b,i}$'s close to 1 can be set to 1.} Thus, we propose to switch off all the antennas for which $a_{b,i}<\epsilon$. After performing the antenna selection, the algorithm needs to be run again with the selected antenna set to find the beamformers with lower dimensions. The proposed joint beamforming and antenna selection method is summarized in Algorithm \ref{algo:iterative}.
\vspace{-2mm}
\textcolor{black}{\subsection{Efficient Implementations of Algorithm \ref{algo:iterative}}}

\subsubsection{\textcolor{black}{Simplified Algorithm}}
It is worth observing that the beamformers produced by the relaxed problem are always feasible for the original problem. This means that it is possible to use the antenna set and the beamformers obtained from the relaxed problem for transmission. Thus, we propose a simpler version of the algorithm, where step 7 is completely ignored. In this case, the choice of $\chi$ becomes more important and with larger $\chi$, the simple algorithm yields closer to the original algorithm, i.e., the achieved solution approaches a Boolean one. In the numerical results, it is illustrated that the beamformers returned by the relaxed problem already yields a good energy efficiency with the good choice of $\chi$. This method is called `Alg. \ref{algo:iterative} \emph{`simple'}' in the numerical results. \textcolor{black}{In Appendix \ref{App3}, we  show how this method reduces the worst-case computational complexity.}

\textcolor{black}{\subsubsection{SOCP Approximation}
Note that the proposed method requires solving a generic non-linear convex program \eqref{eq:EEmax:reform4} in each iteration. It is difficult to solve it efficiently, because it involves the exponential cone. In Appendix \ref{App3}, we show a slightly modified algorithm where the problem at each iteration is an SOCP. This greatly reduces the complexity, since it enables the use of state-of-the-art SOCP solvers such as MOSEK, ECOS, or GUROBI. In the numerical results, we will compare the convergence speed of the solution.}
%\begin{rem}
%Note that the beamformers produced by the relaxed problem are always feasible for the original problem. Thus, in the numerical results, we also illustrate that the beamformers returned by the relaxed problem already yields a good energy efficiency with the good choice of $\alpha$. This method is called `Alg. 1 \emph{`simple'}' in the numerical results.
%\end{rem}
\vspace{-3mm}
\subsection{Initial Points}
Due to the rate constraints, the challenge is to find feasible initial points to run the algorithm presented in the previous section. \textcolor{black}{Note that applying direct convex optimization to find feasible points is not straightforward. One option would be to generate random beamformers until the maximum power constraint and the minimum rate constraints are satisfied. However, this can be very inefficient especially when the rate requirements are high. Another option could be to solve multicell multigroup multicast power minimization problem with minimum SINR constraints (as, e.g. in \cite{Tervo-17MinPower}), which could be equivalently transferred to rate constraint. However, to apply convex optimization in this case, we still need to use semidefinite relaxation, which cannot generally guarantee rank-1 solutions for this problem. Thus, one needs to use Gaussian randomization technique to find feasible rank-1 beamformers. To this end,  we provide herein an initialization method which works effectively for the considered problem.}

\textcolor{black}{The initial $\mathbf{a}^{(0)}$ can be set to all-ones. To find feasible $\mathbf{w}^{(0)},\boldsymbol\beta^{(0)}$, the first observation is that a feasible point of \eqref{eq:EEmax:reform3} is also feasible to \eqref{eq:EEmax:reform4}. That is, we can focus on \eqref{eq:EEmax:reform3} for simplicity to find a feasible initial point. Another observation is that a feasible point for the energy efficiency maximization is also feasible for the sum rate maximization and vice versa. Thus, we can focus on the sum rate maximization problem to simplify the initialization.}
%To find $\mathbf{w}^{(0)},\boldsymbol\beta^{(0)}$, one option is to generate random beamformers $\mathbf{w}^{(0)}$ which satisfy the maximum power constraints, then calculate $\boldsymbol\beta^{(0)}$ by setting \eqref{eq:EEmax:reform2:betaconst} with equality and then solve a feasibility problem
%\begin{subequations}
%\label{eq:EEmax:feas}
%%\begin{align}
%\begin{eqnarray}
%\underset{\mathbf{w},\boldsymbol\gamma,\boldsymbol\beta}{\find} &  & \mathbf{w}, \boldsymbol\beta, \boldsymbol\gamma \\
%\st
%&  &  ||\hat{\mathbf{w}}_{b,i}||_2^2 \leq P_{\text{max}},\; \forall b \in \mathcal{B}, i \in \mathcal{N}_b \label{eq:EE:PC}\\
%%&  & \Phi_{b,i}^{(0)}(\hat{\mathbf{w}}_{b,i}) \geq P_{\text{min}} \label{eq:EEmax:reform3:MinPC} \\
%&  & \gamma_k \leq \Psi_k^{(0)}(\mathbf{w}_{g},\beta_{k}), \forall k \in \mathcal{K}\\
%&  & \eqref{eq:EEmax:reform1:minSINR},\eqref{eq:EEmax:reform2:betaconst}.
%%&  & \gamma_k \geq \bar{\Gamma}_k, \forall k \in \mathcal{K} \\
%%&  & \beta_k \geq {\sigma_{k}^{2} + \sum\limits_{u \in \mathcal{G} \setminus \{g\}} |\mathbf{h}_{b_u,k}\herm \mathbf{w}_{u}|^{2}}, \forall k \in \mathcal{K}.
%\end{eqnarray}
%\end{subequations}
%This can be repeated until a feasible point is found. However, this may be inefficient.
Let us first initialize any $\mathbf{w}^{(0)},\boldsymbol\beta^{(0)}$ and then consider the following problem
\begin{subequations}
\label{eq:EEmax:feas}
%\begin{align}
\begin{eqnarray}
& \hspace{-30pt} \underset{\mathbf{w},\boldsymbol\gamma,\boldsymbol\beta,\mathbf{q},\mathbf{r},\mathbf{p},\boldsymbol\mu}{\max}  & \sum_{g\in\mathcal{G}}r_g-\lambda(\sum_{k\in\mathcal{K}}(q_{1,k}+q_{2,k})\nonumber \\ & & \hspace{-10pt} + \sum_{g\in\mathcal{G}}\mu_g + \sum_{b\in\mathcal{B}}\sum_{i\in\mathcal{N}_b}p_{b,i}) \\
& \hspace{-50pt} \st
&  \hspace{-20pt} ||\hat{\mathbf{w}}_{b,i}||_2^2 - P_{\text{max}} \leq p_{b,i},\; \forall b \in \mathcal{B}, i \in \mathcal{N}_b \label{eq:EE:PC}\\
%&  & \hspace{-20pt} P_{\text{min}} - \Phi_{b,i}^{(0)}(\bar{\mathbf{w}}_{b,i}) \leq p_{b,i},  \forall b \in \mathcal{B}, i \in \mathcal{N}_b\\
%&  & \Phi_{b,i}^{(0)}(\hat{\mathbf{w}}_{b,i}) \geq P_{\text{min}} \label{eq:EEmax:reform3:MinPC} \\
&  & \hspace{-20pt} \gamma_k - \Psi_k^{(0)}(\mathbf{w}_{g},\beta_{k}) \leq q_{1,k}, \forall k \in \mathcal{K}\\
&  & \hspace{-20pt} \max_{k\in\mathcal{K}_g}(\bar{R}_k) - r_g \leq \mu_g , \forall g \in \mathcal{G} \\
&  & \hspace{-20pt} r_g \leq \log(1+\gamma_k), \forall g \in \mathcal{G}, k \in \mathcal{K}_g \\
&  & \hspace{-20pt} {N_0 + \sum\limits_{u \in \mathcal{G} \setminus \{g\}} |\mathbf{h}_{b_u,k}\herm \mathbf{w}_{u}|^{2}} - \beta_k \leq q_{2,k}, \forall k \in \mathcal{K}\\
&  & \hspace{-20pt} q_{1,k} \geq 0, q_{2,k} \geq 0, \forall k \in \mathcal{K}\\
&  & \hspace{-20pt} p_{b,i} \geq 0, \forall b \in \mathcal{B}, i \in \mathcal{N}_b\\
&  & \hspace{-20pt} \mu_{g} \geq 0, \forall g \in \mathcal{G},
\end{eqnarray}
\end{subequations}
where $\lambda$ is a positive penalty parameter and $\mathbf{q} \triangleq \{q_{1,k},q_{2,k}\}_{k\in\mathcal{K}}, \mathbf{p} \triangleq \{p_{b,i}\}_{b\in\mathcal{B},i\in\mathcal{N}_b}, \boldsymbol\mu \triangleq \{\mu_{g}\}_{g\in\mathcal{G}}$ are new slack variables. The above problem is iteratively solved by updating $\Psi_k^{(0)}(\mathbf{w}_{g},\beta_{k})$ after each iteration, until $\mathbf{q}=0,\mathbf{p}=0,\boldsymbol\mu=0$. The feasible point is found very efficiently because the penalty terms are encouraged to zero due to the penalty function in the objective \cite{Vu-16}.

\begin{algorithm}[t]
\begin{algorithmic}[1] \caption{Proposed joint beamforming and antenna selection design.}
\label{algo:iterative} \renewcommand{\algorithmicrequire}{\textbf{Initialization:}}
\REQUIRE Set $n=0$, and generate feasible initial points $(\mathbf{w}^{(n)},\boldsymbol\beta^{(n)},\mathbf{a}^{(n)})$.
%\STATE Perform CC transformation to obtain $(\boldsymbol\phi,\bar{\mathbf{w}}^{(n)},\bar{\boldsymbol\beta}^{(n)})$
\STATEx Phase 1:
\REPEAT
\STATE Solve \eqref{eq:EEmax:reform4} with $(\mathbf{w}^{(n)},\boldsymbol\beta^{(n)},\mathbf{a}^{(n)})$
and denote optimal values as $(\bar{\mathbf{w}}^*,\bar{\boldsymbol\beta}^*,\bar{\mathbf{a}}^*)$.
\STATE Update $\mathbf{w}^{(n+1)} = \bar{\mathbf{w}}^*/\phi ,\boldsymbol\beta^{(n+1)} = \bar{\boldsymbol\beta}^*/\phi,\mathbf{a}^{(n+1)}=\bar{\mathbf{a}}^*/\phi$ and $\Upsilon_{b,i}^{(n+1)}(a_{b,i}), \Psi_k^{(n+1)}(\mathbf{w}_{g},\beta_{k})$.\label{algo:update}
\STATE $n:=n+1$.
\UNTIL convergence
\renewcommand{\algorithmicrequire}{\textbf{Output:}} \REQUIRE $a_{b,i}^* = \tfrac{\bar{a}_{b,i}^*}{\phi^*}, \forall b \in \mathcal{B}, i \in \mathcal{N}_b$
\STATEx Phase 2:
\STATE Set $a_{b,i}=0$, for all $b,i$ for which $a_{b,i}^*<\epsilon$.
\STATE Run steps 1 -- 5 again with fixed $\mathbf{a}$ to find beamformers with reduced dimensions.
\renewcommand{\algorithmicrequire}{\textbf{Output:}} \REQUIRE $\mathbf{w}_{g}^* = \tfrac{\bar{\mathbf{w}}_{g}^*}{\phi^*}, \forall g \in \mathcal{G}$
\end{algorithmic}
\end{algorithm} %\vspace{-1mm}

Algorithm \ref{algo:iterative} is devised based on the combination of continuous relaxation and SCA, which are both suboptimal methods. Thus its performance in comparison with the optimal one is of a great concern. However we remark that analytical investigation of the obtained suboptimal solution to the considered mixed Boolean nonconvex program is challenging, at least given the nonconvexity of the continuous relaxation. Our rationale is that a tight continuous relaxation coupled with the efficacy of the SCA in dealing with nonconvex programs will offer a good performance. To evaluate the performance of Algorithm \ref{algo:iterative} we use numerical experiments which will be presented in Section \ref{sec:NumericalResults}.

\section{Sparsity-based Approach}
\label{sec:Sparse}
Here we propose an alternative formulation based on directly finding a sparse solution for the beamforming vectors, which does not require the introduction of any Boolean variables. Another efficient widely used technique in the mixed-Boolean programming framework is based on sparsity \cite{Bach-12}. Recall that $\hat{\mathbf{w}}_{b,i}$ contains all the coefficients which are related to antenna $i$ of BS $b$. To switch off this antenna, all the elements of $\hat{\mathbf{w}}_{b,i}$ should be zero simultaneously. On the other hand, we want to optimize the number of antennas for a given optimization target.  Let us collect all the beamformers in a matrix $\mathbf{W}\triangleq [\hat{\mathbf{w}}_{1,1},\ldots,\hat{\mathbf{w}}_{1,N_1},\hat{\mathbf{w}}_{2,1},\ldots,\hat{\mathbf{w}}_{2,N_2},\ldots,\hat{\mathbf{w}}_{B,N_B}] \in \mathbb{C}^{\underset{b}{\max}(|\mathcal{G}_b|)\times \sum_{b\in\mathcal{B}}N_b}$. To switch off an antenna, we should set a corresponding column of $\mathbf{W}$ to zero. In other words, to optimize the number of active antennas, we should optimize the number of non-zero columns in $\mathbf{W}$. To this end, we need a group-sparsity technique which promotes the sparsity of the columns of $\mathbf{W}$, but not the rows.   Let us define  $\hat{\mathbf{W}}\triangleq [||\hat{\mathbf{w}}_{1,1}||_2,||\hat{\mathbf{w}}_{1,2}||_2,\ldots,||\hat{\mathbf{w}}_{1,N_1}||_2,||\hat{\mathbf{w}}_{2,1}||_2,\ldots,||\hat{\mathbf{w}}_{2,N_2}||_2, \\ \ldots,||\hat{\mathbf{w}}_{B,N_B}||_2]\trans$. That is, we have calculated the $\ell_2$-norm of each column to equally weight each row in $\hat{\mathbf{w}}_{b,i}$, and avoid row-sparsity. At this point, we note that optimizing the number of non-zero elements in $\hat{\mathbf{W}}$ can be mathematically expressed as $||\hat{\mathbf{W}}||_0$.
% To take these aspects into account, we can use $\ell_1/\ell_2$-regularization method \cite{Bach-12}. That is, we calculate $\ell_2$-norm of each $\hat{\mathbf{w}}_{b,i}$ to equally weight each element of $\hat{\mathbf{w}}_{b,i}$ (because we do not want sparsity for $\hat{\mathbf{w}}_{b,i}$) and then define  $\hat{\mathbf{w}}\triangleq [||\hat{\mathbf{w}}_{1,1}||_2,||\hat{\mathbf{w}}_{1,2}||_2,\ldots,||\hat{\mathbf{w}}_{1,N_1}||_2,||\hat{\mathbf{w}}_{2,1}||_2,\ldots,||\hat{\mathbf{w}}_{2,N_2}||_2, \\ \ldots,||\hat{\mathbf{w}}_{B,N_B}||_2]\trans$. To optimize the number of antennas, we want to optimize the number of non-zero elements in $\hat{\mathbf{w}}$.
As a result, the energy-efficient joint beamforming and antenna selection problem \eqref{EEmax} can be equivalently formulated as
\begin{subequations} \label{EEmax:sp}
\begin{eqnarray}{} \underset{\mathbf{w}}{\maxi} &  & \frac{R(\mathbf{w})}{\sum\limits_{g\in\mathcal{G}}\frac{1}{\eta}||\mathbf{w}_{g}||_2^{2} + P_{\text{RF}}||\hat{\mathbf{W}}||_0  + P_{0}}  \label{eq:sp:EE:obj} \\ \st
 &   & ||\hat{\mathbf{w}}_{b,i}||_2^2 \leq P_{\text{max}},  \forall b \in \mathcal{B}, i \in \mathcal{N}_b, \eqref{eq:SINRconstraints}\label{eq:sp:IntPconstraint}
% &   & ||\hat{\mathbf{w}}_{b,i}||_2^2 \geq a_{b,i}P_{\text{min}},  \forall b \in \mathcal{B}, i \in \mathcal{N}_b\label{eq:MinIntPconstraint}\\
 %&   & ||\hat{\mathbf{w}}_{b,i}||_2^2 \geq a_{b,i}P_{\text{min}},  \forall b \in \mathcal{B}, i \in \mathcal{N}_b \label{eq:MinIntPconstraint} \\
% &   & \log(1+\Gamma_k(\mathbf{w})) \geq {\bar{R}}_k,  \forall k \in \mathcal{K}.\label{eq:sp:SINRconstraints}
 %&   & a_{b,i} \in \{0,1\}, \forall b \in \mathcal{B}, i \in \mathcal{N}_b \label{eq:binary}
\end{eqnarray}
\end{subequations}
Since $||\hat{\mathbf{W}}||_0$ is a discrete function and cannot be optimized as such, some continuous relaxation is required to find a good approximation. Thus, we are interested in solving
\begin{subequations} \label{EEmax:sp:smooth}
\begin{eqnarray}{} \underset{\mathbf{w}}{\maxi} &  & \frac{R(\mathbf{w})}{\sum\limits_{g\in\mathcal{G}}\frac{1}{\eta}||\mathbf{w}_{g}||_2^{2} + (P_{\text{RF}}+\rho)f_i(\mathbf{w})  + P_{0}}  \label{eq:sp:EE:obj} \\ \st
 &   & \eqref{eq:sp:IntPconstraint}
 %&   & a_{b,i} \in \{0,1\}, \forall b \in \mathcal{B}, i \in \mathcal{N}_b \label{eq:binary}
\end{eqnarray}
\end{subequations}
where $f_i(\mathbf{w})$ is some approximation (smoothing function) of $||\hat{\mathbf{W}}||_0$ and $\rho\geq 0$ is an adjustable penalty parameter to control sparsity of $\hat{\mathbf{W}}$. \textcolor{black}{Note that the proposed approach promotes sparsity using the penalty term for the power consumption in the denominator of \eqref{eq:sp:EE:obj}, and not for the whole objective function.}
Next we propose different relaxations $f_i(\mathbf{w})$ to approximate $||\hat{\mathbf{W}}||_0$, a convex one which still maintains the convexity of the denominator in \eqref{eq:sp:EE:obj}, and then two different non-convex smoothing functions requiring additional approximation but yielding better performance.

\subsection{Convex Relaxation of $||\hat{\mathbf{W}}||_0$}

%The closest convex approximation of $||\hat{\mathbf{W}}||_0$ is
%\begin{equation}f_1(\hat{\mathbf{w}})\triangleq \frac{||\hat{\mathbf{W}}||_1}{\sqrt{P_{\text{max}}}}=\sum_{i}\frac{|[\hat{\mathbf{W}}]_i|}{\sqrt{P_{\text{max}}}},\label{f1}\end{equation} where the normalization term $\frac{1}{\sqrt{P_{\text{max}}}}$ makes sure that $\frac{||\hat{\mathbf{W}}||_1}{\sqrt{P_{\text{max}}}}$ is in the range of $[0,1]$ regardless of the transmit power.
The closest convex approximation of $||\hat{\mathbf{W}}||_0$ is $||\hat{\mathbf{W}}||_1$. \textcolor{black}{However, as such, each element of $\hat{\mathbf{W}}$ can be larger than 1, especially when the power constraint $P_{\text{max}}$ is large. By looking at the denominator of \eqref{eq:sp:EE:obj}, we recall that each element of $\hat{\mathbf{W}}$ should indicate whether the corresponding  antenna is  switched on or off. Thus, let us write the normalized form of $\hat{\mathbf{W}}$ as $\hat{\mathbf{W}}_{\text{norm}}\triangleq [\frac{||\hat{\mathbf{w}}_{1,1}||_2}{\sqrt{P_{\text{max}}}},\frac{||\hat{\mathbf{w}}_{1,2}||_2}{\sqrt{P_{\text{max}}}},\ldots,\frac{||\hat{\mathbf{w}}_{1,N_1}||_2}{\sqrt{P_{\text{max}}}},\frac{||\hat{\mathbf{w}}_{2,1}||_2}{\sqrt{P_{\text{max}}}},\ldots, \\ \frac{||\hat{\mathbf{w}}_{2,N_2}||_2}{\sqrt{P_{\text{max}}}}, \ldots,\frac{||\hat{\mathbf{w}}_{B,N_B}||_2}{\sqrt{P_{\text{max}}}}]\trans$. Now the normalization guarantees that each element of $\hat{\mathbf{W}}_{\text{norm}}$ is in the range of $[0,1]$ regardless of the transmit power. Then, $||\hat{\mathbf{W}}||_0$ is approximated as
\begin{equation}f_1(\mathbf{w})\triangleq ||\hat{\mathbf{W}}_{\text{norm}}||_1=\sum_{j=1}^{\sum_{b\in\mathcal{B}}N_b}[\hat{\mathbf{W}}_{\text{norm}}]_j,\label{f1}\end{equation} where $[]_j$ denotes the $j$th element of the argument. In fact, this approach is a slightly modified version of the well-known $\ell_1/\ell_2$-regularization method in \cite{Bach-12}, because now the objective function is not penalized as such, but the approximation function is used in the denominator. Also, it makes sure that each element of $\hat{\mathbf{W}}_{\text{norm}}$ is in the interval $[0,1]$ to approximate the power consumption of the RF chains as accurately as possible.}
%In fact, this approach yields the so called $\ell_1/\ell_2$-regularization method \cite{Bach-12}.
% where $\rho\geq 0$ is an adjustable parameter which promotes sparsity of $\hat{\mathbf{w}}$. A common approach for approximating the discrete function $||\hat{\mathbf{w}}||_0$ in the denominator of \eqref{eq:sp:EE:obj} is to use convex $\ell_1$-norm, yielding the definition of $\ell_1/\ell_2$-regularization method. $\ell_1$-norm is the closest convex approximation of the $\ell_0$-norm.
In this case, we want to solve
\begin{subequations} \label{EEmax:sp:convex}
\begin{eqnarray}{} \underset{\mathbf{w}}{\maxi} &  & \frac{R(\mathbf{w})}{\sum\limits_{g\in\mathcal{G}}\frac{1}{\eta}||\mathbf{w}_{g}||_2^{2} + (P_{\text{RF}}+\rho)||\hat{\mathbf{W}}_{\text{norm}}||_1  + P_{0}}  \label{eq:sp:convex:EE:obj} \\ \st
 &   & \eqref{eq:sp:IntPconstraint}.
% ||\hat{\mathbf{w}}_{b,i}||_2^2 \leq P_{\text{max}},  \forall b \in \mathcal{B}, i \in \mathcal{N}_b\label{eq:sp:convex:IntPconstraint}\\
%% &   & ||\hat{\mathbf{w}}_{b,i}||_2^2 \geq a_{b,i}P_{\text{min}},  \forall b \in \mathcal{B}, i \in \mathcal{N}_b\label{eq:MinIntPconstraint}\\
% %&   & ||\hat{\mathbf{w}}_{b,i}||_2^2 \geq a_{b,i}P_{\text{min}},  \forall b \in \mathcal{B}, i \in \mathcal{N}_b \label{eq:MinIntPconstraint} \\
% &   & \log(1+\Gamma_k(\mathbf{w})) \geq {\bar{R}}_k,  \forall k \in \mathcal{K}.\label{eq:sp:convex:SINRconstraints}
 %&   & a_{b,i} \in \{0,1\}, \forall b \in \mathcal{B}, i \in \mathcal{N}_b \label{eq:binary}
\end{eqnarray}
\end{subequations}
% The normalization term $\frac{1}{\sqrt{P_{\text{max}}}}$ makes sure that $\frac{||\hat{\mathbf{w}}||_1}{\sqrt{P_{\text{max}}}}$ is in the range of $[0,1]$ regardless of the transmit power. We have further added an adjustable penalty parameter $\rho\geq 0$ to control sparsity of $\hat{\mathbf{w}}$.
Due to the convexity of the denominator in \eqref{eq:sp:convex:EE:obj}, the above problem can be solved by following the same transformations as the ones used to arrive at \eqref{eq:EEmax:reform4}, resulting in
\begin{subequations}
\label{eq:sp:EEmax:reform4}
%\begin{align}
\begin{eqnarray}
& \hspace{-20pt} \underset{\phi,\bar{\mathbf{w}},\bar{\boldsymbol\gamma},\bar{\boldsymbol\beta},\bar{\mathbf{r}}}{\maxi} &   \sum_{g\in\mathcal{G}}\bar{r}_g\label{eq:sp:EEmax:reform4:obj}\\
& \hspace{-60pt} \st  & \hspace{-35pt} \sum\limits_{g\in\mathcal{G}}\frac{1}{\eta}\frac{||\bar{\mathbf{w}}_{g}||_2^2}{\phi} + (P_{\text{RF}}+\rho)||\bar{\hat{\mathbf{W}}}_{\text{norm}}||_1 + \phi P_{0} \leq 1  \\
%&  & \hspace{-25pt}  ||\bar{\hat{\mathbf{w}}}_{b,i}||_2^2 \leq \bar{a}_{b,i}\bar{v}_{b,i},\; \forall b \in \mathcal{B}, i \in \mathcal{N}_b \label{eq:EE:PC}\\
&  & \hspace{-35pt} \frac{||\bar{\hat{\mathbf{w}}}_{b,i}||_2^2}{\phi} \leq \phi P_{\text{max}}, \forall b \in \mathcal{B}, i \in \mathcal{N}_b \label{eq:sp:EEmax:reform4:PC}\\
&   & \hspace{-35pt} \eqref{eq:EEmax:reform4:SINRapprox}-\eqref{eq:EEmax:reform4:rateconst}
%&  & \hspace{-35pt} \bar{v}_{b,i} \leq \phi P_{\text{max}}, \forall b \in \mathcal{B}, i \in \mathcal{N}_b \\
%&  & \hspace{-25pt} \Phi_{b,i}^{(n)}(\bar{\hat{\mathbf{w}}}_{b,i},\phi) \geq \phi(\frac{\bar{a}_{b,i}}{\phi})^\alpha P_{\text{min}}, \forall b \in \mathcal{B}, i \in \mathcal{N}_b \label{eq:EEmax:reform2:MinPC} \\
%&  & \hspace{-25pt} \bar{v}_{b,i} \geq \bar{a}_{b,i} P_{\text{min}}, \forall b \in \mathcal{B}, i \in \mathcal{N}_b \label{eq:EEmax:reform2:vMinPC} \\
%&  & \hspace{-35pt} 0 \leq \bar{a}_{b,i} \leq \phi, \forall b \in \mathcal{B}, i \in \mathcal{N}_b\\
% &  & \hspace{-35pt} \bar{\gamma}_k \leq \phi\Psi_k^{(n)}(\frac{\bar{\mathbf{w}}_{g}}{\phi},\frac{\bar{\beta}_{k}}{\phi}), \forall k \in \mathcal{K}\\
% %&  & \hspace{-35pt} \phi \Phi_{b,i}^{(n)}(\frac{\hat{\mathbf{w}}_{b,i}}{\phi}) \geq (\bar{a}_{b,i} -\bar{\rho}_{b,i}) P_{\text{min}}, \forall b \in \mathcal{B}, i \in \mathcal{N}_b \label{eq:EEmax:reform4:MinPC} \\
% &  & \hspace{-35pt} \bar{r}_g \geq \phi \max_{k\in\mathcal{K}_g}(\bar{R}_k), \forall g \in \mathcal{G} \\
% &  & \hspace{-35pt} \phi\bar{\beta}_k \geq {\phi^2 \sigma_{k}^{2} + \sum\limits_{u \in \mathcal{G} \setminus \{g\}} |\mathbf{h}_{b_u,k}\herm \bar{\mathbf{w}}_{u}|^{2}}, \forall k \in \mathcal{K}\\
% &  & \hspace{-35pt} \bar{r}_g \leq \phi\log(1+\frac{\bar{\gamma}_k}{\phi}), \forall g \in \mathcal{G}, k \in \mathcal{K}_g.
\end{eqnarray}
\end{subequations}
which is solved iteratively until convergence. After this, the antennas for which $||\hat{\mathbf{w}}_{b,i}||_2/\sqrt{P_{\text{max}}} < \epsilon$, where $\epsilon$ is a small threshold are set to zero and the algorithm is rerun for the chosen antenna sets.
%For the sake of completeness, we summarize the method in Algorithm \ref{algo:sp}. The same SOCP approximations presented in \ref{SOC} similarly apply to Algorithm \ref{algo:sp}.

\subsection{Non-Convex Relaxation of $||\hat{\mathbf{W}}||_0$}
Although being simple and widely used, a convex $\ell_1$-norm relaxation may not yield the best solution, and its efficiency to provide a sparse enough solution is highly dependent on the penalty parameter. Bearing this in mind, here we propose another relaxation based on a non-convex smoothing function. Let us define a convex function $\varphi_{b,i}(\hat{\mathbf{w}}_{b,i})\triangleq \frac{||\hat{\mathbf{w}}_{b,i}||_2}{\sqrt{P_{\text{max}}}}$. Then, we propose the following two alternative functions to approximate $||\hat{\mathbf{W}}||_0$
% \begin{equation}
% f(\hat{\mathbf{w}})=\sum\limits_{i=1}^{\sum_{b\in\mathcal{B}}N_b}\frac{\frac{([\hat{\mathbf{w}}]_i)^2}{P_{\text{max}}}}{\xi+\frac{([\hat{\mathbf{w}}]_i)^2}{P_{\text{max}}}},
% \end{equation}
\begin{subequations}
\begin{eqnarray}
f_2(\mathbf{w})=\sum\limits_{b\in\mathcal{B}}\sum\limits_{i\in\mathcal{N}_b}\varphi_{b,i}(\hat{\mathbf{w}}_{b,i})^{\frac{1}{\varsigma}} \label{f2}\\
f_3(\mathbf{w})=\sum\limits_{b\in\mathcal{B}}\sum\limits_{i\in\mathcal{N}_b}\log_2(1+\varphi_{b,i}(\hat{\mathbf{w}}_{b,i})^{\frac{1}{\varsigma}})\label{f3}
\end{eqnarray}
\end{subequations}
where $\varsigma \geq 1$ is a parameter controlling the steepness of the curve of the smoothing function. That is, larger $\varsigma$ means steeper curve.
%$\xi$ is a small regularization parameter to avoid numerical instability.
We have that $\varphi_{b,i}(\hat{\mathbf{w}}_{b,i})^{\frac{1}{\varsigma}} \to 0$ when $\frac{||\hat{\mathbf{w}}_{b,i}||_2}{\sqrt{P_{\text{max}}}}\to 0$ and $\varphi_{b,i}(\hat{\mathbf{w}}_{b,i})^{\frac{1}{\varsigma}} \to 1$  when $\frac{||\hat{\mathbf{w}}_{b,i}||_2}{\sqrt{P_{\text{max}}}} \to 1$, for every $\varsigma\geq 1$. Now, $f_2(\mathbf{w})$ and $f_3(\mathbf{w})$ are both concave for every $\varphi_{b,i}\geq 0, \varsigma \geq 1$.
% When $\xi>0$, $f(\hat{\mathbf{w}})$ is concave for $[\hat{\mathbf{w}}]_i\geq \sqrt{\frac{1}{3}\xi}$ which can be proved by calculating the point where the second derivative of $f(\hat{\mathbf{w}})$ turns to negative.
Problem \eqref{EEmax:sp} is now approximated as
\begin{subequations} \label{EEmax:sp:nonconvex}
\begin{eqnarray}{} \underset{\mathbf{w}}{\maxi} &  & \frac{R(\mathbf{w})}{\sum\limits_{g\in\mathcal{G}}\frac{1}{\eta}||\mathbf{w}_{g}||_2^{2} + P_{\text{RF}}f_i(\mathbf{w}) + \rho f_i(\mathbf{w})  + P_{0}}  \label{eq:sp:nonconvex:EE:obj} \\ \st
 &   & \eqref{eq:sp:IntPconstraint}
% ||\hat{\mathbf{w}}_{b,i}||_2^2 \leq P_{\text{max}},  \forall b \in \mathcal{B}, i \in \mathcal{N}_b\label{eq:sp:nonconvex:IntPconstraint}\\
%% &   & ||\hat{\mathbf{w}}_{b,i}||_2^2 \geq a_{b,i}P_{\text{min}},  \forall b \in \mathcal{B}, i \in \mathcal{N}_b\label{eq:MinIntPconstraint}\\
% %&   & ||\hat{\mathbf{w}}_{b,i}||_2^2 \geq a_{b,i}P_{\text{min}},  \forall b \in \mathcal{B}, i \in \mathcal{N}_b \label{eq:MinIntPconstraint} \\
% &   & \log(1+\Gamma_k(\mathbf{w})) \geq {\bar{R}}_k,  \forall k \in \mathcal{K}.\label{eq:sp:nonconvex:SINRconstraints} %\\
 %&  & [\hat{\mathbf{w}}]_i\geq \sqrt{\frac{1}{3}\xi}.
 %&   & a_{b,i} \in \{0,1\}, \forall b \in \mathcal{B}, i \in \mathcal{N}_b \label{eq:binary}
\end{eqnarray}
\end{subequations}
Compared to previous formulation, the concavity of $f_i(\mathbf{w})$ makes the denominator nonconvex. Thus, we introduce an affine function $\hat{f}_i^{(n)}(\mathbf{w})$ as the first-order Taylor approximation of $f_i(\mathbf{w})$ at point $\mathbf{w}^{(n)}$. Following again the same idea as in \eqref{eq:sp:EEmax:reform4}, we iteratively solve
\begin{subequations}
\label{eq:sp:EEmax:reform5}
%\begin{align}
\begin{eqnarray}
& \hspace{-20pt} \underset{\phi,\bar{\mathbf{w}},\bar{\boldsymbol\gamma},\bar{\boldsymbol\beta},\bar{\mathbf{r}}}{\maxi} &   \sum_{g\in\mathcal{G}}\bar{r}_g\label{eq:sp:EEmax:reform4:obj}\\
& \hspace{-60pt} \st  & \hspace{-35pt} \sum\limits_{g\in\mathcal{G}}\frac{1}{\eta}\frac{||\bar{\mathbf{w}}_{g}||_2^2}{\phi} + (P_{\text{RF}}+\rho)\hat{f}_i^{(n)}(\frac{\bar{\mathbf{w}}}{\phi}) + \phi P_{0} \leq 1  \\
&  & \hspace{-35pt} \frac{||\bar{\hat{\mathbf{w}}}_{b,i}||_2^2}{\phi} \leq \phi P_{\text{max}}, \forall b \in \mathcal{B}, i \in \mathcal{N}_b \label{eq:sp:EEmax:reform4:PC}\\
&   & \hspace{-35pt} \eqref{eq:EEmax:reform4:SINRapprox}-\eqref{eq:EEmax:reform4:rateconst}.
\end{eqnarray}
\end{subequations}
For the sake of completeness, we summarize the sparsity-based methods in Algorithm \ref{algo:sp}. The same SOCP approximations presented in Appendix \ref{App3} similarly apply to Algorithm \ref{algo:sp}.

\begin{algorithm}[t]
\begin{algorithmic}[1] \caption{Proposed sparsity-based joint beamforming and antenna selection design.}
\label{algo:sp} \renewcommand{\algorithmicrequire}{\textbf{Initialization:}}
\REQUIRE Set $n=0$, and generate feasible initial points $(\mathbf{w}^{(n)},\boldsymbol\beta^{(n)})$.
%\STATE Perform CC transformation to obtain $(\boldsymbol\phi,\bar{\mathbf{w}}^{(n)},\bar{\boldsymbol\beta}^{(n)})$
\STATEx Phase 1:
\REPEAT
\STATE Solve \eqref{eq:sp:EEmax:reform4} (or \eqref{eq:sp:EEmax:reform5}) for $f_1(\mathbf{w})$ (or $f_2(\mathbf{w}),f_3(\mathbf{w})$) with $(\mathbf{w}^{(n)},\boldsymbol\beta^{(n)})$
and denote optimal values as $(\bar{\mathbf{w}}^*,\bar{\boldsymbol\beta}^*)$.
\STATE Update $\mathbf{w}^{(n+1)} = \bar{\mathbf{w}}^*/\phi ,\boldsymbol\beta^{(n+1)} = \bar{\boldsymbol\beta}^*/\phi$ and $\Psi_k^{(n+1)}(\mathbf{w}_{g},\beta_{k})$ for $f_1(\mathbf{w})$ (and $\hat{f}_i^{(n+1)}(\mathbf{w})$ for $f_2(\mathbf{w}),f_3(\mathbf{w})$).\label{algo:update}
\STATE $n:=n+1$.
\UNTIL convergence
\renewcommand{\algorithmicrequire}{\textbf{Output:}} \REQUIRE $\mathbf{w}_{g}^* = \tfrac{\bar{\mathbf{w}}_{g}^*}{\phi^*}, \forall g \in \mathcal{G}$
\STATEx Phase 2:
\STATE Set $\hat{\mathbf{w}}_{b,i}=0$, for all $b,i$ for which $||\hat{\mathbf{w}}_{b,i}^*||_2/\sqrt{P_{\text{max}}}<\epsilon$.
\STATE Run steps 1-5 again with the chosen antenna set to find beamformers with reduced dimensions.
\renewcommand{\algorithmicrequire}{\textbf{Output:}} \REQUIRE $\mathbf{w}_{g}^* = \tfrac{\bar{\mathbf{w}}_{g}^*}{\phi^*}, \forall g \in \mathcal{G}$
\end{algorithmic}
\end{algorithm} %\vspace{-1mm}

\section{Energy Efficiency and Sum Rate Trade-Offs}\label{sec:tradeoff}

Here we consider the trade-off between the energy efficiency and sum rate maximization. First, we propose a new optimization metric to which all the optimization algorithms developed in the previous section are applicable as such. Secondly, an alternative formulation based on a scalarization approach of the multi-objective optimization problem is proposed.

\subsection{Power-Weighted Energy Efficiency Maximization}

The power-weighted EE maximization problem is stated as
\begin{subequations} \label{PWEE:EEmax}
\begin{eqnarray}{} \underset{\mathbf{w},\mathbf{a}}{\maxi} &  & \frac{R(\mathbf{w})}{\kappa g(\mathbf{w},\mathbf{a}) + P_{0}}  \label{eq:PWEE:obj} \\ \st
 &   & ||\hat{\mathbf{w}}_{b,i}||_2^2 \leq a_{b,i}P_{\text{max}},  \forall b \in \mathcal{B}, i \in \mathcal{N}_b, \eqref{eq:SINRconstraints} \label{eq:PWEE:IntPconstraint}\\
% &   & ||\hat{\mathbf{w}}_{b,i}||_2^2 \geq a_{b,i}P_{\text{min}},  \forall b \in \mathcal{B}, i \in \mathcal{N}_b\label{eq:MinIntPconstraint}\\
 %&   & ||\hat{\mathbf{w}}_{b,i}||_2^2 \geq a_{b,i}P_{\text{min}},  \forall b \in \mathcal{B}, i \in \mathcal{N}_b \label{eq:MinIntPconstraint} \\
 %&   & \eqref{eq:SINRconstraints} \\
 &   & a_{b,i} \in \{0,1\}, \forall b \in \mathcal{B}, i \in \mathcal{N}_b \label{eq:PWEE:binary}
\end{eqnarray}
\end{subequations}
where $g(\mathbf{w},\mathbf{a})\triangleq \sum\limits_{g\in\mathcal{G}}\frac{1}{\eta}||\mathbf{w}_{g}||_2^{2} + P_{\text{RF}}\sum\limits_{b\in\mathcal{B}}\sum\limits_{i\in\mathcal{N}_b}a_{b,i}$ is a function denoting the adjustable power consumption.
%On the other hand, \eqref{eq:MinIntPconstraint} sets the minimum antenna-specific power constraint, i.e., if $a_{b,i}=1$, then the transmit power from antenna $i$ has to be at least $P_{\text{min}}$ and if $a_{b,i}=0$, the beamformers can be zero.
In the proposed metric, $\kappa \in [0,1]$ is a fixed parameter to control the weighting between the energy efficiency and the sum rate maximization.

The intuition behind the proposed metric is that EE can be adjusted by changing the power which then affects the achievable sum rate eventually.  More explicitly, when the weighting factor for the power consumption is decreased, the power consumption becomes less significant and thus the PWEE metric aims to increase the rate. In the considered problem formulation, the power can be adjusted both in terms of transmit power and RF chain power. When $\kappa=0$, the denominator becomes a fixed value $P_{0}$ and the problem is equivalent to maximizing the sum rate. {Similarly, when $\kappa$ is increased, it means that there is a penalty for increasing the transmit power, which results in a decrease in the sum rate. } In particular, when $\kappa=1$, the problem is equivalent to energy efficiency maximization. In summary, by varying $\kappa$ from 0 to 1, we can study the trade-off between EE and SR.

It is worth mentioning that if antenna selection is not considered, the weighting is set to transmit power only, because then the RF chain power is not adjusted. Due to the fact that $\kappa$ is fixed, the problem can be solved using exactly the same algorithms as developed in the previous section, by just adding the weight in front of the function $g(\cdot)$. In the numerical results, we use this approach together with the continuous relaxation approach, i.e., it is named as \emph{Alg. 1, PWEE}.

\begin{rem}
Conventionally, the trade-off problem has been treated by calculating the weighted sum of EE and SE \cite{Tang-14,Amin-16,He-13EESE}. The key difference in the proposed formulation is that the weighting is only for adjustable power, which means that exactly the same derived algorithm for energy efficiency maximization can also solve the trade-off problem.
\end{rem}

\subsection{Scalarization Approach}
Here we show an alternative formulation for the EE-SR trade-off problem. We present the algorithm framework for the mixed-Boolean programming based formulation, but the method can be applied to the sparsity-based formulations as well. We focus on the EE and the sum rate trade-off problem
\begin{subequations} \label{EESEmax1}
\begin{eqnarray}{} \underset{\mathbf{w},\mathbf{a}}{\maxi} &  &  [\frac{R(\mathbf{w})}{g(\mathbf{w},\mathbf{a}) + P_0},R(\mathbf{w})]  \label{eq:EEsc:obj} \\ \st
 &   & \eqref{eq:PWEE:IntPconstraint}, \eqref{eq:PWEE:binary} \label{all}.
% ||\hat{\mathbf{w}}_{b,i}||_2^2 \leq a_{b,i}P_{\text{max}},  \forall b \in \mathcal{B}, i \in \mathcal{N}_b\label{eq:sc:IntPconstraint}\\
% %&   & ||\hat{\mathbf{w}}_{b,i}||_2^2 \geq a_{b,i}P_{\text{min}},  \forall b \in \mathcal{B}, i \in \mathcal{N}_b\label{eq:MinIntPconstraint}\\
% %&   & ||\hat{\mathbf{w}}_{b,i}||_2^2 \geq a_{b,i}P_{\text{min}},  \forall b \in \mathcal{B}, i \in \mathcal{N}_b \label{eq:MinIntPconstraint} \\
% &   & \log(1+\Gamma_k(\mathbf{w})) \geq {\bar{R}}_k,  \forall k \in \mathcal{K}, \label{eq:SINRconstraints:scalari}\\
% &   & a_{b,i} \in \{0,1\}, \forall b \in \mathcal{B}, i \in \mathcal{N}_b.\label{eq:sc:Binconstraint}
\end{eqnarray}
\end{subequations}
%\begin{subequations} \label{EESEmax}
%\begin{eqnarray}{} \underset{\mathbf{w},\mathbf{a}}{\maxi} &  & [\frac{R(\mathbf{w})}{g(\mathbf{w},\mathbf{a})}, \varsigma R(\mathbf{w})]  \label{eq:EE:obj} \\ \st
% &   & ||\hat{\mathbf{w}}_{b,i}||_2^2 \leq a_{b,i}P_{\text{max}},  \forall b \in \mathcal{B}, i \in \mathcal{N}_b\label{eq:IntPconstraint}\\
% &   & ||\hat{\mathbf{w}}_{b,i}||_2^2 \geq a_{b,i}P_{\text{min}},  \forall b \in \mathcal{B}, i \in \mathcal{N}_b\label{eq:MinIntPconstraint}\\
% %&   & ||\hat{\mathbf{w}}_{b,i}||_2^2 \geq a_{b,i}P_{\text{min}},  \forall b \in \mathcal{B}, i \in \mathcal{N}_b \label{eq:MinIntPconstraint} \\
% &   & \Gamma_k(\mathbf{w}) \geq \bar{\Gamma}_k,  \forall k \in \mathcal{K}, \label{eq:SINRconstraints}\\
% &   & a_{b,i} \in \{0,1\}, \forall b \in \mathcal{B}, i \in \mathcal{N}_b
%\end{eqnarray}
%\end{subequations}
The above problem is a multi-objective optimization problem with two conflicting objectives. A common method to solve this type of problem is the use a scalarization approach \cite{Ngatchou-05}. Therein, the problem is transformed to a single-objective optimization problem
\begin{subequations} \label{EESEmax2}
\begin{eqnarray}{} \underset{\mathbf{w},\mathbf{a}}{\maxi} &  &  \varrho\frac{R(\mathbf{w})}{g(\mathbf{w},\mathbf{a})+P_0} + (1-\varrho)R(\mathbf{w})  \label{eq:SEEE:obj} \\ \st
&   & \eqref{all}
%\eqref{eq:sc:IntPconstraint}-\eqref{eq:sc:Binconstraint}
% &   & ||\hat{\mathbf{w}}_{b,i}||_2^2 \leq a_{b,i}P_{\text{max}},  \forall b \in \mathcal{B}, i \in \mathcal{N}_b\label{eq:sc1:IntPconstraint}\\
%% &   & ||\hat{\mathbf{w}}_{b,i}||_2^2 \geq a_{b,i}P_{\text{min}},  \forall b \in \mathcal{B}, i \in \mathcal{N}_b\label{eq:MinIntPconstraint}\\
% %&   & ||\hat{\mathbf{w}}_{b,i}||_2^2 \geq a_{b,i}P_{\text{min}},  \forall b \in \mathcal{B}, i \in \mathcal{N}_b \label{eq:MinIntPconstraint} \\
% &   & \log(1+\Gamma_k(\mathbf{w})) \geq {\bar{R}}_k,  \forall k \in \mathcal{K}, \label{eq:sc1:SINRconstraints}\\
% &   & a_{b,i} \in \{0,1\}, \forall b \in \mathcal{B}, i \in \mathcal{N}_b\label{eq:sc1:Binconstraint}
\end{eqnarray}
\end{subequations}
where $\varrho \in [0,1]$ is a fixed parameter to control the weighting between energy efficiency and sum rate. However, the problem in this specific formulation is that the units of the two objectives are inconsistent and the numerical values are not comparable. To make the problem tractable, we use similar approach as that in \cite{Tang-14} to formulate the problem as
\begin{subequations} \label{EESEmax3}
\begin{eqnarray}{} \underset{\mathbf{w},\mathbf{a}}{\maxi} &  &  \frac{R(\mathbf{w})}{g(\mathbf{w},\mathbf{a})+P_0} + \varrho \frac{R(\mathbf{w})}{P_{\text{min}}}  \label{eq:SEEE:obj} \\ \st
&   & \eqref{all} %\eqref{eq:sc:IntPconstraint}-\eqref{eq:sc:Binconstraint}
% &   & ||\hat{\mathbf{w}}_{b,i}||_2^2 \leq a_{b,i}P_{\text{max}},  \forall b \in \mathcal{B}, i \in \mathcal{N}_b\label{eq:sc1:IntPconstraint}\\
% %&   & ||\hat{\mathbf{w}}_{b,i}||_2^2 \geq a_{b,i}P_{\text{min}},  \forall b \in \mathcal{B}, i \in \mathcal{N}_b\label{eq:MinIntPconstraint}\\
% %&   & ||\hat{\mathbf{w}}_{b,i}||_2^2 \geq a_{b,i}P_{\text{min}},  \forall b \in \mathcal{B}, i \in \mathcal{N}_b \label{eq:MinIntPconstraint} \\
% &   & \log(1+\Gamma_k(\mathbf{w})) \geq {\bar{R}}_k,  \forall k \in \mathcal{K}, \label{eq:sc1:SINRconstraints}\\
% &   & a_{b,i} \in \{0,1\}, \forall b \in \mathcal{B}, i \in \mathcal{N}_b
\end{eqnarray}
\end{subequations}
where $P_{\text{min}}$ is the minimum power which is known to be consumed for transmission and $\varrho\geq 0$ is the weighting parameter. The difference of the above formulation to \cite{Tang-14} is that here we use $P_{\text{min}}$ instead of maximum possible power consumption $P_{\text{tot}}$ to scale the data rate in the objective. The reason for this is that the EE is defined both in terms of transmit power and RF chain power, which means that if $P_{\text{tot}}$ is used, the value of $\varrho \frac{R(\mathbf{w})}{P_{\text{tot}}}$ would be clearly smaller than $\frac{R(\mathbf{w})}{g(\mathbf{w},\mathbf{a})+P_0}$ even with quite large a value of $\varrho$. Thus, the objective would focus on EE and the range of $\varrho$ would be difficult to define to reasonably exploit the trade-off. Note that $P_{\text{min}}$ can be defined depending on the number of user groups, i.e., we can set $P_{\text{min}}=P_0+P_{\text{RF}}\sum_{b\in\mathcal{B}}X_b$ according to Remark \ref{rem:minAnt}. In the numerical results, we will show that the above formulation achieves a nice trade-off curve mostly in the range of $\varrho \in [0,1]$.

Because all the other constraints remain the same, they can be handled as shown previously. However, the objective function \eqref{eq:SEEE:obj} is not a conventional fractional function anymore. Thus, we equivalently reformulate it as
\begin{subequations}
\label{eq:EESEmax:reform2}
%\begin{align}
\begin{eqnarray}
\hspace{-10pt} \underset{\mathbf{w},\boldsymbol\gamma, \mathbf{v},\mathbf{a},\mathbf{r},x}{\maxi} &  &  x + \varrho \frac{\sum_{g\in\mathcal{G}}r_g}{P_{\text{min}}}\\
\st
&    & \hspace{-10pt} g(\mathbf{v},\mathbf{a}) + P_{0} \leq \frac{\sum_{g\in\mathcal{G}}r_g}{x} \label{eq:EESEmax:reform2:ncvx} \\
%&    & t \leq \varsigma \sum_{g\in\mathcal{G}}r_g \\
%&  & \hspace{-10pt} ||\hat{\mathbf{w}}_{b,i}||_2^2 \leq a_{b,i}^\alpha v_{b,i},\; \forall b \in \mathcal{B}, i \in \mathcal{N}_b \label{eq:EEmax:reform1:MaxPC}\\
%&  & \hspace{-10pt} v_{b,i} \leq P_{\text{max}}, \forall b \in \mathcal{B}, i \in \mathcal{N}_b \label{eq:EEmax:reform1:vmin}\\
%&   & \hspace{-10pt} ||\hat{\mathbf{w}}_{b,i}||_2^2 \geq a_{b,i}P_{\text{min}},  \forall b \in \mathcal{B}, i \in \mathcal{N}_b\label{eq:reform1:MinIntPconstraint}\\
%&  & \hspace{-10pt} v_{b,i} \geq a_{b,i} P_{\text{min}}, \forall b \in \mathcal{B}, i \in \mathcal{N}_b \label{eq:EEmax:reform1:vMinPC} \\
%%&  & \hspace{-10pt} ||\hat{\mathbf{w}}_{b,i}||_2^2 \geq a_{b,i}^\alpha P_{\text{min}}, \forall b \in \mathcal{B}, i \in \mathcal{N}_b \label{eq:EEmax:reform1:MinPC} \\
%%&  & \hspace{-10pt} v_{b,i} \geq a_{b,i} P_{\text{min}}, \forall b \in \mathcal{B}, i \in \mathcal{N}_b \label{eq:EEmax:reform1:vMinPC} \\
%&  & \hspace{-10pt} a_{b,i} \in \{0,1\}, \forall b \in \mathcal{B}, i \in \mathcal{N}_b \label{eq:EEmax:reform1:binary}\\
%&  & \hspace{-10pt} \gamma_k \leq \frac{|\mathbf{h}_{b_g,k} \mathbf{w}_{g}|^{2}}{{\sigma_{k}^{2} + \sum\limits_{u \in \mathcal{G} \setminus \{g\}} |\mathbf{h}_{b_u,k} \mathbf{w}_{u}|^{2}}}, \forall k \in \mathcal{K}\label{eq:EEmax:reform1:rate}\\
%&  & \hspace{-10pt} \gamma_k \geq \bar{\Gamma}_k, \forall k \in \mathcal{K}\label{eq:EEmax:reform1:minSINR}\\
%&  & \hspace{-10pt} r_g \leq \log(1+\gamma_k), \forall g \in \mathcal{G}, k \in \mathcal{K}_g\label{eq:EESEmax:reform1:weakestRATE}\\
&  & \hspace{-10pt} \eqref{eq:EEmax:reform0:MaxPC}, \eqref{eq:EEmax:reform0:vmin}, \eqref{eq:EEmax:reform1:rate}, \eqref{eq:EEmax:reform1:minSINR}, \eqref{eq:EEmax:reform1:weakestRATE}
\end{eqnarray}
\end{subequations}
where $x$ is a new variable denoting the total energy efficiency.
To tractably reformulate non-convex constraint \eqref{eq:EESEmax:reform2:ncvx}, we replace it with the following two constraints
\begin{subequations}
\label{eq:EESEmax:reform3}
%\begin{align}
\begin{eqnarray}
g(\mathbf{v},\mathbf{a}) + P_0 \leq \frac{r^2}{x} \\
\sum_{g\in\mathcal{G}}r_g \geq r^2
\end{eqnarray}
\end{subequations}
where $r$ is a new variable representing the square root of the total sum rate.
As a result, we can express \eqref{EESEmax3} as
\begin{subequations}
\label{eq:EESEmax:reform3}
%\begin{align}
\begin{eqnarray}
 \underset{\mathbf{w},\boldsymbol\gamma, \mathbf{v},\mathbf{a},\boldsymbol\beta,\mathbf{r}, x,r}{\maxi} &  &   x + \varrho \frac{\sum_{g\in\mathcal{G}}r_g}{P_{\text{min}}} \\
 \st  & &  g(\mathbf{v},\mathbf{a}) + P_0 \leq \frac{r^2}{x} \label{eq:EESEmax:NewApprox}\\
&    &  r^2 \leq \sum_{g\in\mathcal{G}}r_g  \label{eq:EESEmax:reform3:cvx1} \\
%&    & \hspace{-45pt} x^2 \leq \sum_{g\in\mathcal{G}}r_g  \label{eq:EESEmax:reform3:cvx2} \\
&  &   \frac{||\hat{\mathbf{w}}_{b,i}||_2^2}{v_{b,i}} \leq a_{b,i}^\chi ,\; \forall b \in \mathcal{B}, i \in \mathcal{N}_b \label{eq:EESEmax:reform2:MaxPC}\\
%&  & \hspace{-45pt}  t^2 \leq \sum_{g\in\mathcal{G}}r_g \\
%&  & v_{b,i} \leq P_{\text{max}}\\
%&  & \hspace{-15pt} ||\hat{\mathbf{w}}_{b,i}||_2^2 \geq a_{b,i}^\alpha P_{\text{min}}, \forall b \in \mathcal{B}, i \in \mathcal{N}_b \label{eq:EEmax:reform2:MinPC} \\
%&   & \hspace{-45pt} ||\hat{\mathbf{w}}_{b,i}||_2^2 \geq a_{b,i}P_{\text{min}},  \forall b \in \mathcal{B}, i \in \mathcal{N}_b\label{eq:reform1:MinIntPconstraintSE}\\
&  &  \gamma_k \leq \frac{|\mathbf{h}_{b_g,k}\herm \mathbf{w}_{g}|^{2}}{\beta_k}, \forall k \in \mathcal{K} \label{eq:EEmax:reform2:quadoverlinSE}\\
&  &  0 \leq a_{b,i} \leq 1, \forall b \in \mathcal{B}, i \in \mathcal{N}_b\\
%&  & \gamma_k \geq \bar{\Gamma}_k \\
%&  & \hspace{-15pt} \beta_k \geq {\sigma_{k}^{2} + \sum\limits_{u \in \mathcal{G} \setminus \{g\}} |\mathbf{h}_{b_u,k}\herm \mathbf{w}_{u}|^{2}}, \forall k \in \mathcal{K}\label{eq:EEmax:reform2:betaconst}\\
&  &  \eqref{eq:EEmax:reform0:MaxPC}, \eqref{eq:EEmax:reform0:vmin}, \eqref{eq:EEmax:reform1:minSINR}, \eqref{eq:EEmax:reform1:weakestRATE}, \eqref{eq:rep:reform2:betaconst}.
\end{eqnarray}
\end{subequations}
Now we can see that the objective function is affine and all the other constraints except \eqref{eq:EESEmax:reform2:MaxPC}, \eqref{eq:EEmax:reform2:quadoverlinSE} and newly introduced constraint \eqref{eq:EESEmax:NewApprox} are convex. Using again the idea of SCA, \eqref{eq:EEmax:reform2:quadoverlinSE} and \eqref{eq:EESEmax:reform2:MaxPC} can be approximated as in \eqref{eq:EEmax:quad_over_lin_app} and \eqref{eq:EEmax:Approx}, respectively. To approximate the right side of \eqref{eq:EESEmax:NewApprox}, we can write
\begin{eqnarray}\label{eq:EESEmax:New_app}
\frac{r^2}{x} \geq \frac{2r^{(n)}}{x^{(n)}}r
 - (\frac{r^{(n)}}{x^{(n)}})^2 x \triangleq \Delta^{(n)}(r,x).
\end{eqnarray}
With the approximations \eqref{eq:EEmax:quad_over_lin_app}, \eqref{eq:EEmax:Approx}, and \eqref{eq:EESEmax:New_app}, we solve the following convex problem at iteration $n$ of the SCA method
\begin{subequations}
\label{eq:EESEmax:reform4}
%\begin{align}
\begin{eqnarray}
& \hspace{-35pt} \underset{\mathbf{w},\boldsymbol\gamma, \mathbf{v},\mathbf{a},\boldsymbol\beta,\mathbf{r},r,x}{\maxi} &   x + \varrho \frac{\sum_{g\in\mathcal{G}}r_g}{P_{\text{min}}} \\
& \hspace{-20pt} \st  & \hspace{0pt} g(\mathbf{v},\mathbf{a}) + P_0 \leq \Delta^{(n)}(r,x)\\
&  &  \hspace{0pt} \frac{||\hat{\mathbf{w}}_{b,i}||_2^2}{v_{b,i}} \leq \Upsilon_{b,i}^{(n)}(a_{b,i}) ,\; \forall b \in \mathcal{B}, i \in \mathcal{N}_b \label{eq:EESEmax:reform4:MaxPC}\\
%&  & v_{b,i} \leq P_{\text{max}}\\
%&  & \hspace{-15pt} ||\hat{\mathbf{w}}_{b,i}||_2^2 \geq a_{b,i}^\alpha P_{\text{min}}, \forall b \in \mathcal{B}, i \in \mathcal{N}_b \label{eq:EEmax:reform2:MinPC} \\
%&   & \hspace{-40pt} \Phi_{b,i}^{(n)}(\hat{\mathbf{w}}_{b,i}) \geq (a_{b,i}-\rho_{b,i})P_{\text{min}},  \forall b \in \mathcal{B}, i \in \mathcal{N}_b\label{eq:reform4:MinIntPconstraintSE}\\
&  & \hspace{0pt} \gamma_k \leq \Psi_{k}^{(n)}(\mathbf{w}_{g},\beta_k), \forall k \in \mathcal{K} \label{eq:EEmax:reform4:quadoverlinSE}\\
%&  & \hspace{-45pt} t \leq \varsigma \sum_{g\in\mathcal{G}}r_g  \label{eq:EESEmax:reform5:cvx1} \\
&    & \hspace{0pt} r^2 \leq \sum_{g\in\mathcal{G}}r_g  \label{eq:EESEmax:reform5:cvx2} \\
&  & \hspace{0pt} 0 \leq a_{b,i} \leq 1, \forall b \in \mathcal{B}, i \in \mathcal{N}_b\\
%&  & \gamma_k \geq \bar{\Gamma}_k \\
%&  & \hspace{-15pt} \beta_k \geq {\sigma_{k}^{2} + \sum\limits_{u \in \mathcal{G} \setminus \{g\}} |\mathbf{h}_{b_u,k}\herm \mathbf{w}_{u}|^{2}}, \forall k \in \mathcal{K}\label{eq:EEmax:reform2:betaconst}\\
&  & \hspace{0pt} \eqref{eq:EEmax:reform0:MaxPC}, \eqref{eq:EEmax:reform0:vmin}, \eqref{eq:EEmax:reform1:minSINR}, \eqref{eq:EEmax:reform1:weakestRATE}, \eqref{eq:rep:reform2:betaconst}.
\end{eqnarray}
\end{subequations}
The proposed algorithm is summarized in Algorithm \ref{algo:scalarization} for the sake of completeness. We can apply the same method to the sparsity-based formulation as well.

\begin{rem}
When $\varrho=0$ (i.e., the problem reduces to the energy efficiency maximization problem), the method derived above is an alternative solution for Algorithm \ref{algo:iterative}.
\end{rem}
We have introduced two different formulations to investigate the EE-SR trade-off problem. They come from two different views. The scalarization method is more explicit as it weights the two metrics directly, while the PWEE is more implicit by weighting the power consumption to adjust the trade-off. However we remark that both formulations offer the same performance as demonstrated in the next section.

\begin{algorithm}[t]
\begin{algorithmic}[1] \caption{Proposed joint beamforming and antenna selection design for energy efficiency and sum rate trade-off.}
\label{algo:scalarization} \renewcommand{\algorithmicrequire}{\textbf{Initialization:}}
\REQUIRE Set $n=0$, and generate feasible initial points $(\mathbf{w}^{(n)},\boldsymbol\beta^{(n)},\mathbf{a}^{(n)},r^{(n)},x^{(n)})$.
%\STATE Perform CC transformation to obtain $(\boldsymbol\phi,\bar{\mathbf{w}}^{(n)},\bar{\boldsymbol\beta}^{(n)})$
\REPEAT
\STATE Solve \eqref{eq:EESEmax:reform4} with $(\mathbf{w}^{(n)},\boldsymbol\beta^{(n)},\mathbf{a}^{(n)},r^{(n)},x^{(n)})$
and denote optimal values as $(\mathbf{w}^*,\boldsymbol\beta^*,\mathbf{a}^*,r^*,x^*)$.
\STATE Update $\mathbf{w}^{(n+1)} = \mathbf{w}^* ,\boldsymbol\beta^{(n+1)} = \boldsymbol\beta^*,\mathbf{a}^{(n+1)}=\mathbf{a}^*,r^{(n+1)}=r^*,x^{(n+1)}=x^*$ and $\Upsilon_{b,i}^{(n+1)}(a_{b,i}), \Psi_k^{(n+1)}(\mathbf{w}_{g},\beta_{k}), \Delta^{(n+1)}(r,x)$.\label{algo:update}
\STATE $n:=n+1$.
\UNTIL convergence
\renewcommand{\algorithmicrequire}{\textbf{Output:}} \REQUIRE $a_{b,i}^*, \forall b \in \mathcal{B}, i \in \mathcal{N}_b$
\STATE Set $a_{b,i}=0$, for all $b,i$ for which $a_{b,i}^*<\epsilon$.
\STATE Run steps 1-5 again with fixed $\mathbf{a}$ to find beamformers with reduced dimensions.
\renewcommand{\algorithmicrequire}{\textbf{Output:}} \REQUIRE $\mathbf{w}_{g}^*, \forall g \in \mathcal{G}$
\end{algorithmic}
\end{algorithm} %\vspace{-1mm}

\section{Numerical Results}
\label{sec:NumericalResults}
% * <namletran@gmail.com> 2017-09-27T16:46:34.896Z:
%
% There is still space so you can tabulate all the parameters.
%
% ^.
% We evaluate the performance for a quasistatic frequency flat Rayleigh fading channel model where the distance from the BSs to all the users is $d$ m. This models the worst-case interference scenario where all the cells are on top of each other to fully account for the effect of inter-cell interference.
The performances of the developed algorithms are evaluated in a quasistatic frequency flat Rayleigh fading channel model. \textcolor{black}{We model a scenario with $B=2$ adjacent cells, where all the users are between the two BSs to account for the most severe inter-cell interference situation, as illustrated in Fig. \ref{fig: SimulationModel}. In all the other figures except \ref{fig: UnEqEEvsN} and \ref{fig: ImperfectCSI}, we set the distance from the BSs to all the users to $d=250$ meters ($d$ is the cell radius), and randomly assign them to multicasting groups of equal size. This models a scenario where the groups are very close to each other and all of the users have strong average interference. In Figures \ref{fig: UnEqEEvsN} and \ref{fig: ImperfectCSI}, where the methods are compared with existing schemes, we consider more typical scenarios where the users are randomly dropped between the cells so that the interference profiles between the users are different.}
The path loss is calculated as $30\log_{10}(d_{b,k})+35$ dB, where $d_{b,k}$ is the distance from BS $b$ to user $k$. Each of the BSs serves equal number of $G_b=U$ randomly assigned user groups with $|\mathcal{K}_g|=L$ users per group, i.e., the total number of users in the network is $K=BUL$. We assume a bandwidth of 20 MHz and noise power is set to $N_0=-125$ dBW. The antenna specific maximum power constraints are assumed to be equal for all the antennas over the whole bandwidth.  The fixed simulation parameters are summarized in Table \ref{Tab. 1} while the other parameters are given in the figures.
 %and the algorithms have been stopped when the change of the objective value has been smaller than $\xi=10^{-6}$ between two last iterations.
%  We set $\eta=0.35, \epsilon=0.001, P_{\text{max}}=1$ W, $P_{\text{RF}}=0.4$ W, $P_{\text{sta}}=4.5$ W, $P_{\text{UE}}=0.1 W$, and the other simulation parameters are given in the figures.

\begin{figure}[t]
\centering
  \includegraphics[width=0.7\columnwidth]{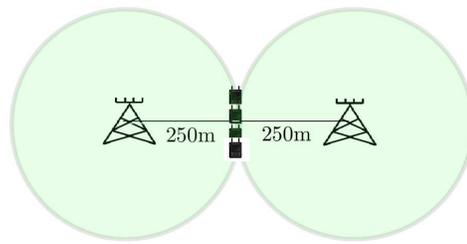}
  \caption{Simulation model.}
  \label{fig: SimulationModel}
\end{figure}

\begin{table}[tb]
\caption{Simulation Parameters}

\centering{}%
\begin{tabular}{c|c}
\hline
{Parameters}  & {Value}\tabularnewline
\hline
\hline
Path loss & $30\log_{10}\left(d\text{ [m]}\right)+35$ [dB]\tabularnewline
Cell radius  & $d=250\text{ [m]}$\tabularnewline
Active RF chain-specific power $P_{\text{RF}}$  & 0.4 W\tabularnewline
Static power consumption $P_{\text{sta}}$  & 4.5 W\tabularnewline
Per-user power consumption $P_{\text{UE}}$  & 0.1 W\tabularnewline
Power amplifier efficiency $\eta$  & 0.35\tabularnewline
Maximum antenna-specific power $P_{\text{max}}$  & 1 W\tabularnewline
Number of BSs $B$  & 2\tabularnewline
Number of groups per cell $G_{b}$  & $U$\tabularnewline
Number of users per group $K_{g}$  & $L$\tabularnewline
Number of Tx antennas $N_b$  & $N$\tabularnewline
Signal bandwidth $W$  & 20 MHz\tabularnewline
Noise power $N_0$ & -125 dBW\tabularnewline
$\epsilon$ & 0.001 \tabularnewline
\hline
\end{tabular}\vspace{-2mm}
\label{Tab. 1}
\end{table}

%\subsection{Results Without Minimum Power Constraints}
% * <namletran@gmail.com> 2017-09-27T16:49:14.879Z:
%
% You may include the convergence of Algorithm 2.
%
% ^.
%Fig. \ref{fig: Convergence} illustrates the convergence of the relaxed problem and the achieved energy efficiency for two different values of $\chi$, and two random channel realizations. The same initial points have been used for both values of $\chi$. We can see that the convergence speed is fast in the considered settings for both cases. However, we observe that they converge to different solutions in both examples. Specifically, the objective value of the relaxed problem after convergence is higher for $\chi=1$, but the achieved energy efficiency (after recovering the Boolean solution) is worse than with $\chi=2$. This is because with $\chi=1$, more antenna selection variables are non-Boolean after convergence, which results in a worse antenna selection result. Another observation is that with $\chi=2$ (which achieves better energy efficiency), the objective value returned by the relaxed problem and the achieved energy efficiency are very close to each other. This means that the solution of the relaxed problem is already very close to Boolean. The better solution is achieved with only slightly decreased convergence speed of the relaxed problem. The examples demonstrate the effectiveness of the proposed formulation. The impact of $\chi$ on the average energy efficiency is studied in the next experiment.

\textcolor{black}{\subsection{Algorithm performance}}

Fig. \ref{fig: Convergence} illustrates the average convergence of the relaxed problem \textcolor{black}{(i.e., phase 1) and the achieved EE (after phase 2) for Alg. \ref{algo:iterative} and Alg. \ref{algo:sp}. First in Fig. \ref{fig:Convergence1}, we have run Alg. \ref{algo:iterative} with $\chi=1$ and $\chi=2$. The same initial points have been used for both values of $\chi$. We can see that the convergence speed in phase 1 is fast in the considered setting for both cases. However, we observe that they converge to different solutions. Specifically, the objective value of the relaxed problem (phase 1) after convergence is higher for $\chi=1$, but the achieved EE (after phase 2)} is worse than with $\chi=2$. This is because with $\chi=1$, more antenna selection variables are non-Boolean after convergence, which results in a worse antenna selection result. Another observation is that with $\chi=2$ (which achieves better energy efficiency), the objective value returned by the relaxed problem and the achieved EE are very close to each other. This means that the solution of the relaxed problem is already very close to Boolean. The better solution is achieved with only slightly decreased convergence speed of the relaxed problem. The examples demonstrate the effectiveness of the proposed formulation. The impact of $\chi$ on the average EE is studied in the next experiment. \textcolor{black}{Fig. \ref{fig:Convergence2} then illustrates the convergence of Alg. \ref{algo:sp} for different relaxations. We can observe that both smoothing functions provide fast convergence. Note that the methods converge to different objective values, because the curves show the convergence of phase 1. It is observed that compared to Alg. \ref{algo:iterative}, the relaxation is very loose, but it still results in the same average EE after phase 2.}

\begin{figure}[t]
\centering
\subfigure[Convergence of Alg. \ref{algo:iterative}]{\label{fig:Convergence1}\includegraphics[width=0.75\columnwidth]{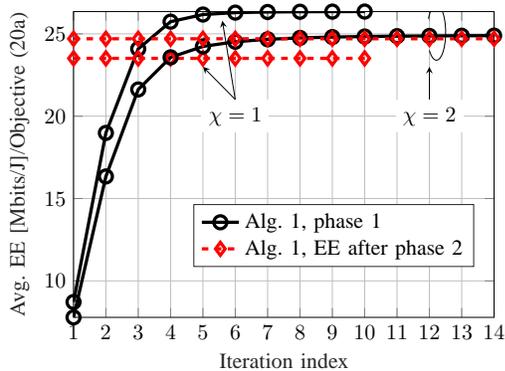}}
%\subfigure[Channel realization 2]{\label{fig:Convergence2}\includegraphics[width=1\columnwidth]{Convergence2}}
\subfigure[Convergence of Alg. \ref{algo:sp}]{\label{fig:Convergence2}\includegraphics[width=0.75\columnwidth]{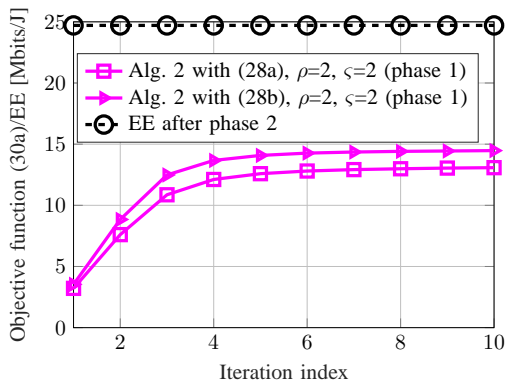}}
%    \begin{subfigure}[b]{0.3\textwidth}
%        \includegraphics[width=\textwidth]{Figures_Nats/EEmaxConvAndCDFNatsSub1}
%        %\caption{Convergence illustration for three different initial points}
%        %\label{fig:EEmaxConvAndCDFNatsSub1}
%    \end{subfigure}
%    ~
%    \begin{subfigure}[b]{0.3\textwidth}
%        \includegraphics[width=\textwidth]{Figures_Nats/EEmaxConvAndCDFNatsSub2}
%        %\caption{Average cumulative distribution function}
%        %\label{fig:EEmaxConvAndCDFNatsSub2}
%    \end{subfigure}
%\caption{Convergence of the relaxed problem and achieved energy efficiency for two different channel realizations with $N=$16, $L=$2, $U=$2, $\bar{R}=$ 46.4 Mbits/s, $d=250$ m. The flat lines denote the achieved energy efficiencies after calculating the beamformers for the chosen antenna sets (i.e., after terminating Alg. \ref{algo:iterative}).}
\caption{\textcolor{black}{Average convergence of the relaxed problem (phase 1) and achieved energy efficiency (after phase 2) of Alg. \ref{algo:iterative} and Alg. \ref{algo:sp} $N=$16, $L=$2, $U=$2, $\bar{R}=$ 46.4 Mbits/s. The flat lines denote the achieved energy efficiencies after phase 2.}}
\label{fig: Convergence}
\end{figure}

%\begin{figure}[t]
%\centering
%  \includegraphics[width=0.83\columnwidth]{Figures/Convergence}
%  \caption{Convergence of the relaxed problem and achieved energy efficiency for two different channel realizations with $P_{\text{RF}}=$1 W, $N=$16, $\bar{\Gamma}=$ 0 dB. The flat lines denote the achieved energy efficiencies after calculating the beamformers for the chosen antenna sets (i.e., after terminating Alg. \ref{algo:iterative}).}
%  \label{fig: Convergence}
%\end{figure}

\begin{figure}[t]
\centering
  \includegraphics[width=0.75\columnwidth]{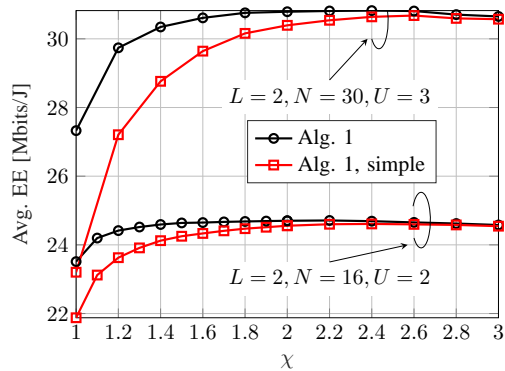}
  \caption{Average energy efficiency versus $\chi$ with $\bar{R}=$46.4 Mbits/s.}
  \label{fig: EEvsAlpha}
\end{figure}

% Fig. \ref{fig: ConvergenceSparse} illustrates the convergence of Alg. \ref{algo:sp} for different relaxations. We can observe that all the methods have fast convergence. Note that the methods converge to different objective value, because the curves only show the convergence of steps 1-5. However, after running the whole algorithm, the achieved energy efficiencies are very close to each other as demonstrated in the next experiments.

%\begin{figure}[t]
%\centering
%  \includegraphics[width=1\columnwidth]{ConvergenceSparse}
%  \caption{Achieved objective value for two different random channel realizations (denoted with solid and dashed lines) with $L=2, U=2, N=16$, $\bar{R}=$20 Mbits/s, $d=250$ m.}
%  \label{fig: ConvergenceSparse}
%\end{figure}

\begin{figure}[t]
\centering
  \includegraphics[width=0.8\columnwidth]{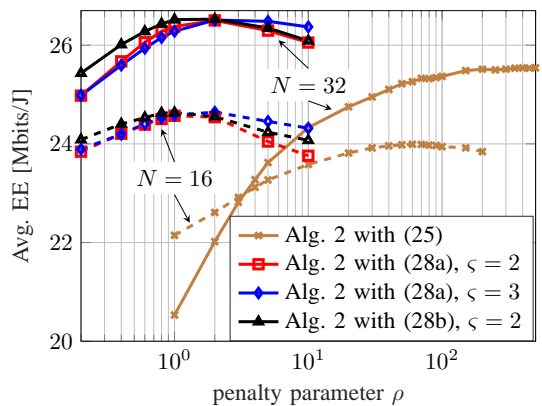}
  \caption{The effect of $\rho$ on the average energy efficiency of sparsity-based methods with $L=2$, $U=2$, $\bar{R}=$20 Mbits/s. In solid lines, $N=30$ and $N=16$ for dashed lines.}
  \label{fig: EEvsRho}
\end{figure}

%\begin{figure}[t]
%\centering
%  \includegraphics[width=0.83\columnwidth]{Figures/EEvsPdyn}
%  \caption{Energy efficiency versus $P_{\text{RF}}$ with $N=$16, $\bar{\Gamma}=$ 0 dB.}
%  \label{fig: EEvsPdyn}
%\end{figure}

\begin{figure}[t]
\centering
  \includegraphics[width=0.75\columnwidth]{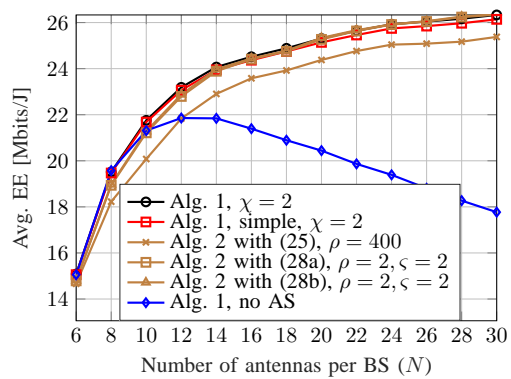}
  \caption{Average energy efficiency versus $N$ with $L=2$, $U=2$, $\bar{R}=$20 Mbits/s.}
  \label{fig: EEvsN}
\end{figure}

Fig. \ref{fig: EEvsAlpha} illustrates the effect of $\chi$ on the average energy efficiency with different simulation parameters. We also illustrate the performance of the simplified algorithm (Alg. 1 \emph{'simple'}), where the beamformers achieved from the relaxed problem (i.e., step 7 is ignored in Alg. 1) are used for transmission. We can see that the choice of $\chi$ affects the achieved energy efficiency, and the choice of the best $\chi$ depends on the system parameters. In these cases, $\chi=2.2-2.4$ gives the best performance. More importantly, the choice of $\chi$ significantly affects the performance of the simple method, implying that it is enough to use the beamformers and antennas based on the relaxed problem for transmission. This reduces the computational load because step 7 can be ignored. The fact that the performance of the simple method is very close to the original method means that the solution of the relaxed method is very close to Boolean with the correct choice of $\chi$. We also observe a saturation effect for $\chi$ implying that there exists a trade-off between the tightness of the continuous relaxation and the loss of optimality with the proposed algorithm.

%Fig. \ref{fig: ConvergenceSparse} illustrates the convergence of Alg. \ref{algo:sp} for different relaxations. We can observe that all the methods have fast convergence. Note that the methods converge to different objective values, because the curves only show the convergence of steps 1--5. However, after running the whole algorithm, the achieved energy efficiencies are very close to each other as demonstrated in the next experiments.

Fig. \ref{fig: EEvsRho} demonstrates the effect of penalty parameter $\rho$ on the average energy efficiency of Alg. \ref{algo:sp} when using different smoothing functions to approximate $||\hat{\mathbf{W}}||_0$. First, we observe that the convex relaxation \eqref{f1} is not efficient in promoting sparsity because very high penalty parameter is required to yield the best energy efficiency. The performance is also highly dependent on the choice of $\rho$, making it difficult to estimate good $\rho$. Moreover, the achieved performance is still inferior to the other schemes, motivating the use of non-convex relaxations. We can see that all the non-convex relaxations result in approximately the same performance with the correct choice of $\rho$. Also, a good performance is achieved even for $\rho=0$, implying that the methods are efficient in promoting sparsity. It is observed that the relaxation \eqref{f3} is the best option in promoting sparsity with very small values of $\rho$.
% We can also see that the case of $\varsigma=3$ is the best one for the smallest $\rho$, because it means steeper curve, and, thus, is more efficient in terms of sparsity.
% We can see that it has a significant effect on the performance and the more available antennas, the larger $\rho$ is required to get sparse enough solution.

%\begin{figure}[t]
%\centering
%  \includegraphics[width=0.83\columnwidth]{EEvsN}
%  \caption{Average energy efficiency versus $N$ with $\bar{R}=$20 Mbits/s.}
%  \label{fig: EEvsN}
%\end{figure}

%Fig. \ref{fig: EEvsPdyn} illustrates the average energy efficiency versus $P_{\text{RF}}$. Specifically, the proposed algorithm is run with $\alpha=1$ and $\alpha=1.5$, and compared with the method where only beamformers are optimized (Alg. 1, no AS). First, we can observe that Alg. 1 with $\alpha=1.5$ achieves the best results, and the gap between $\alpha=1$ and  $\alpha=1.5$ increases with $P_{\text{RF}}$. It is interesting to observe that with $\alpha=1.5$, even the beamformers obtained from the relaxed problem result in very good performance, which again verifies that the relaxed problem yields near-binary solutions. However, the simple method with $\alpha=1$ is even worse than the method without antenna selection for larger $P_{\text{RF}}$. Finally, we can see that the proposed algorithm provides significant energy efficiency gains over the method without antenna selection (i.e., roughly 5-50\% with the considered system parameters).

Fig. \ref{fig: EEvsN} shows the average EE versus the number of antennas per BS. First, we see that EE starts to decrease without antenna selection when $N>12$. On the other hand, significant gains are achieved with the proposed algorithms and the gains naturally increase with the number of antennas. We can see that the simplified version of Alg. \ref{algo:iterative} is very close to the original method even when the number of antennas is large, which again motivates the effectiveness of the proposed formulation. On the other hand, as already observed in Fig. \ref{fig: EEvsN}, Alg. \ref{algo:sp} with \eqref{f1} gives clearly worse performance than the other JBAS schemes do. However, Alg. \ref{algo:sp} with \eqref{f2} and \eqref{f3} give the performance very close to Alg. \ref{algo:iterative}, although slightly worse when the number of antennas is small. The reason for this is the fixed penalty parameter, which is now optimized for the larger number of antennas, resulting in too sparse a solution with smaller numbers of antennas. This is further investigated in Fig. \ref{fig: nActive}. \textcolor{black}{The reason why \eqref{f1} is clearly worse is that it approximates $||\hat{\mathbf{W}}||_0$ by a pretty flat function, which is not very close to $||\hat{\mathbf{W}}||_0$. As it is generally known, the convex approximations (like the $\ell_1$-norm approximation) are easier to solve, but  not efficient in promoting sparsity. The non-convex ones are far better in this regard, since they approximate $||\hat{\mathbf{W}}||_0$ by a tighter function which has a shape closer to $||\hat{\mathbf{W}}||_0$. Thus, the nonconvex approximation is simply a better approximation of $||\hat{\mathbf{W}}||_0$, resulting in an improved performance.}

% * <namletran@gmail.com> 2017-09-27T16:52:40.496Z:
%
% Space is still there so zoom up all the figures to fit the length.
%
% ^.

Fig. \ref{fig: TXpowervsN} illustrates the transmit powers versus $N$ with Alg. \ref{algo:iterative} and a method without antenna selection, with the same simulation parameters as those in Fig. \ref{fig: EEvsN}. We can see that without antenna selection, it is energy-efficient to increase total transmit power when the number of antennas increases. This is achieved by decreasing the average transmit power per antenna. However, with JBAS, both the total transmit power and transmit power per antenna decrease when $N$ increases in the considered setting, \textcolor{black}{until it starts to saturate. The reason  is that when the number of antennas grows large enough, increasing the number of antennas for a fixed number of users does not provide additional SE gain, i.e., we end up with choosing the same number of antennas on average. This saturation of the number of the active antennas can be observed in Fig. \ref{fig: nActive}. Increasing the number of antennas enables reducing the per-antenna power, meaning that  the power amplifiers for each antenna can be cheaper.} It can be concluded that the EE gains of the JBAS compared to the method without antenna selection are achieved by switching off some of the antennas but using larger per-antenna transmit power.

\begin{figure}[t]
\centering
  \includegraphics[width=0.75\columnwidth]{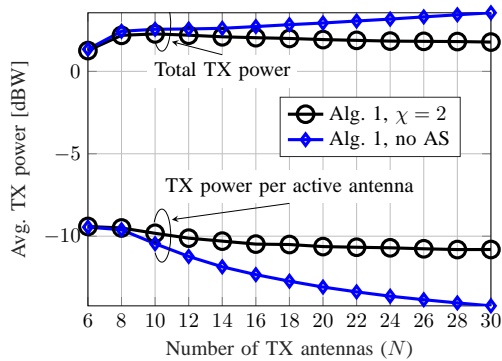}
  \caption{Total transmit power and per-antenna power versus $N$ with $U=$3, $L=$2, $\bar{R}=$ 20 Mbits/s.}
  \label{fig: TXpowervsN}
\end{figure}

Fig. \ref{fig: nActive} displays the average number of active antennas to maximize the energy efficiency versus $N$ with the same simulation parameters as those in Fig. \ref{fig: EEvsN}. It is observed that the more antennas available, the more active antennas are chosen for energy-efficient transmission. This makes sense, because when there are more antennas, there is more spatial diversity. Thus, sum rate gains due to better beamforming overwhelm the increased power consumption of activating some antennas. \textcolor{black}{Also, the optimal number of antennas saturates to a certain value when $N$ grows large, because when the number of antennas is large enough, increasing the number of antennas for a fixed number of users  provides little SE gain.} As can be seen, Alg. 2 with \eqref{f1} switches off too many antennas when the number of available antennas is small, and, on the other hand, not enough when $N$ is large. This also explains the performance gap in Fig. \ref{fig: EEvsN}. In this case, $\rho$ has been chosen to give the best result for $N=30$, which is clearly too large, when $N$ is smaller. Although being a simple method, the optimal penalty parameter has to be carefully chosen. By looking at the performance of the other schemes, we can see that the average number of active antennas for the best energy efficiency saturates to 17--18 in the considered setting. \textcolor{black}{It can be concluded that when the number of users is fixed, it is energy-efficient to use a relatively small number of active antennas if the number of available antennas is much larger than the number of users in the system.}

\begin{figure}[t]
\centering
  \includegraphics[width=0.75\columnwidth]{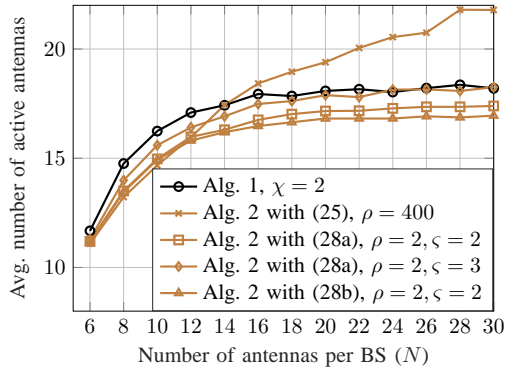}
  \caption{The average number of active antennas to maximize the energy efficiency versus $N$ with $U=$3, $L=$2, $\bar{R}=$ 20 Mbits/s.}
  \label{fig: nActive}
\end{figure}

% Fig. \ref{fig: EEvsGroupsize} shows the impact of increasing the group size on the energy efficiency performance of Alg. \ref{algo:iterative}. Obviously, the gains of antenna selection decrease with the increasing group size, since wider beams are required to serve all the users in the groups, and the resulting inter-group interference can be better addressed with higher numbers of active antennas. It is also observed that the energy efficiency decreases with increasing group size due to the same reason. Once again we can see that the simple method is very close to the original method even when the group size changes.
% \begin{figure}[t]
% \centering
%   \includegraphics[width=1\columnwidth]{EEvsGroupsize}
%   \caption{Average EE versus the group size $L$ with $U=2$, $\bar{R}=$ 20 Mbits/s, $\chi=2$, $d=250$ m.}
%   \label{fig: EEvsGroupsize}
% \end{figure}

%The effect of the distance between the BSs and users for the methods is plotted in Fig. \ref{fig: EEvsDistance}. Obviously, the performance gains of joint beamforming and antenna selection are larger when the distance is smaller, due to the fact that the users' rate requirements can be satisfied with a smaller number of antennas. When the distance is small, it is better to use larger $\chi$.
%
%\begin{figure}[t]
%\centering
%  \includegraphics[width=1\columnwidth]{EEvsDistance}
%  \caption{Average EE versus the distance with $L=2$, $U=3$, $N=30$, $\bar{R}=$ 46.4 Mbits/s.}
%  \label{fig: EEvsDistance}
%\end{figure}

Fig. \ref{fig: SOCPconv} illustrates the convergence of the SOCP approximation algorithm presented in Appendix \ref{App3}. The examples are for two different channel realizations and the same initial beamformers have been used for both algorithms. We can see that there is no significant difference in the convergence speed between the methods.

\begin{figure}[t]
\centering
  \includegraphics[width=0.75\columnwidth]{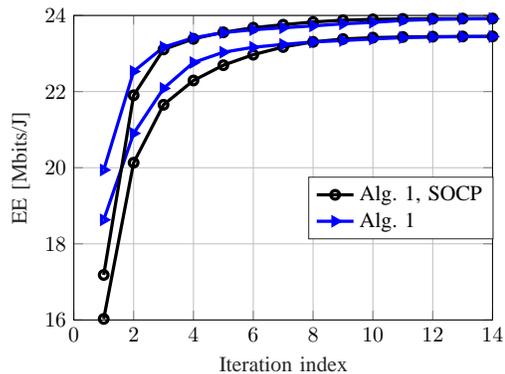}
  \caption{Convergence illustration of the SOCP approximation algorithm with $L=2$, $U=2$, $N=12$, $\bar{R}=$ 20 Mbits/s.}
  \label{fig: SOCPconv}
\end{figure}

\textcolor{black}{\subsection{Comparison to Other Schemes}}

\textcolor{black}{Here the proposed method is compared with some existing methods in the literature. Now the users are assumed to have random locations between the two BSs, and the results are averaged over user locations and channel realizations. One benchmark method is the multiuser MISO JBAS scheme in \cite{He-16JBAS}, where the SCA with semidefinite programming formulation and the $\ell_1/\ell_\infty$ norm is used to solve the JBAS problem. This method can be straightforwardly extended   to the multi-cell multiuser scenario.
% , because single-cell transmission (with orthogonal resources or without interference coordination) is always significantly worse than the multi-cell cooperation.
Basically this means that compared to our multigroup multicast transmission, each user is served by different stream, i.e., it uses resources $L$ times more inefficiently, but can use better user-specific beamformers. Similarly we  also extend the method in \cite{He-16JBAS} to multi-cell multigroup multicasting scheme, i.e., to solve exactly the same problem. The second method in comparison is the multi-cell single-group multicasting scheme in \cite{He-15}, where the user groups inside the same cell are served by the orthogonal resources, but the neighboring cells reuse the same resources similarly. We also use the multi-cell single-group multicasting scheme using our proposed JBAS method.}

\textcolor{black}{In Fig. \ref{fig: UnEqEEvsN}, the EE is plotted for different numbers of antennas. First, it is observed that the proposed methods are significantly superior to all the other schemes. As expected, the extension of JBAS method in \cite{He-16JBAS} provides the closest performance to the proposed schemes. However, its inferiority to Alg. \ref{algo:iterative} is in line with the results observed from Fig. \ref{fig: EEvsN} already, where  Alg. \ref{algo:sp} with convex approximation provides significantly worse performance than the other schemes. Note that the method in \cite{He-16JBAS} uses convex $\ell_1/\ell_\infty$-norm, and uses semidefinite relaxation to approximate the EE problem, and then use SCA for further approximation. Thus, its complexity is higher \cite{Ben-Tal-01}, because the semidefinite relaxation dramatically increases the problem size, and the approximation is not so accurate because rank-1 solutions cannot be guaranteed \cite{Tervo-17MinPower}.}

\textcolor{black}{In Fig. \ref{fig: ImperfectCSI}, we then investigate the effect of imperfect CSI on the average performance of the methods. The imperfect CSI is modeled so that the BSs have knowledge of the noisy channel $\hat{\mathbf{h}}=\mathbf{h}+\tilde{\mathbf{h}}$, where $\mathbf{h}$ is the perfect channel and $\tilde{\mathbf{h}}$ is zero-mean complex Gaussian noise with variance $\sigma_e^2$ per element. It is observed that the performance degrades when the accuracy of the channel decreases, as expected. However, Alg. \ref{algo:iterative} still achieves clearly the best performance, and JBAS is still very useful for EE improvements. It is also interesting to note that the performance of single-group transmission mode becomes closer to multigroup mode when the accuracy decreases. The reason for this is that there is less interference due to orthogonal transmission inside the cell, which makes the interference due to imperfect channel less significant.}

\begin{figure}[t]
\centering
  \includegraphics[width=0.75\columnwidth]{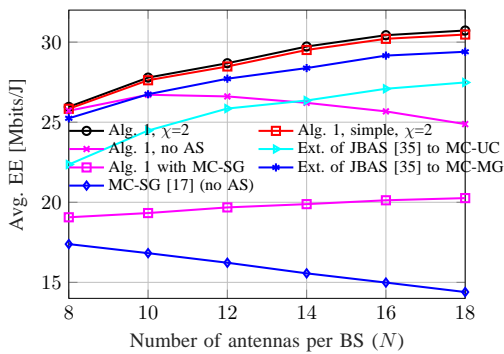}
  \caption{EE comparison to existing schemes $L=2$, $U=2$, $\bar{R}=$ 20 Mbits/s, $d=250$ m with different number of antennas.}
  \label{fig: UnEqEEvsN}
\end{figure}

\begin{figure}[t]
\centering
  \includegraphics[width=0.75\columnwidth]{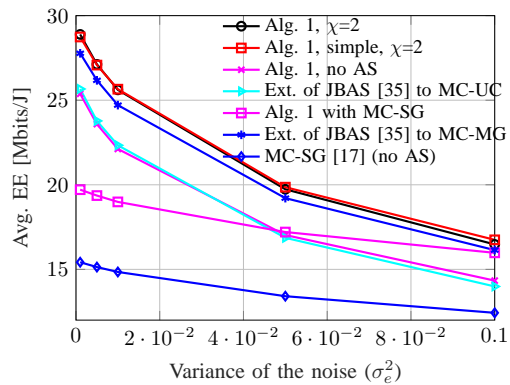}
  \caption{The effect of imperfect CSI on the EE performance with $L=2$, $U=2$, $\bar{R}=$ 20 Mbits/s.}
  \label{fig: ImperfectCSI}
\end{figure}

\textcolor{black}{\subsection{EE-SR Tradeoff}}

Fig. \ref{fig: TradeOff} plots the average EE-SR trade-off curve. The trade-off curve has been simulated by sweeping the parameter $\kappa$ from 0 to 1 in 'Alg. 1, PWEE` and $\varrho$ from 0 to 5 in Alg. 3. It is observed that a lot wider trade-off region is achieved with joint beamforming and antenna selection. It is also observed that both of the proposed algorithms give the same average performance.

\begin{figure}[t]
\centering
  \includegraphics[width=0.75\columnwidth]{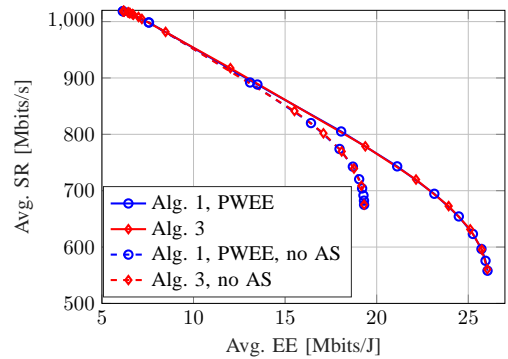}
  \caption{Average EE-SE trade-off curve $L=2$, $U=2$, $N=24$, $\chi=2$, $\bar{R}=$ 20 Mbits/s.}
  \label{fig: TradeOff}
\end{figure}

Fig. \ref{fig: nActiveEE} shows the average number of active antennas in the trade-off curve. More specifically, we can see the number of active antennas which gives certain energy efficiency. In the considered setting, activating all the antennas gives the worst energy efficiency (this is the point where the sum rate is maximized). On the other hand, the energy efficiency maximizing number of active antennas is approximately 18, which yields the minimum sum rate. \textcolor{black}{From Figs. \ref{fig: TradeOff} and \ref{fig: nActiveEE}, we can find the following
observation. Looking at the EE-maximizing point of the method without AS, the
JBAS scheme can maintain the same average sum rate as the scheme without AS, with more than 25\% increase in the EE. This gain is achieved by switching off approximately half of the RF chains.}
%It is worth mentioning that in general, the sum rate maximizing point is not the energy efficiency minimizing point.

\begin{figure}[t]
\centering
  \includegraphics[width=0.75\columnwidth]{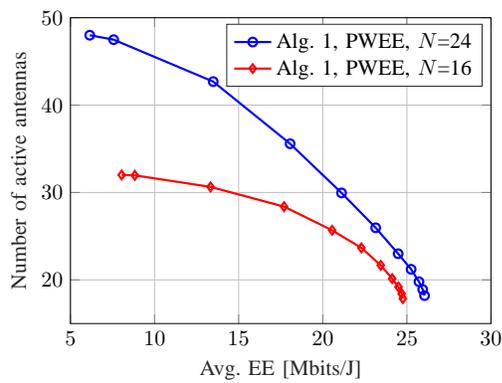}
  \caption{The average number of active antennas to achieve a certain average energy efficiency with $L=2$, $U=2$, $\chi=2$, $\bar{R}=$ 20 Mbits/s.}
  \label{fig: nActiveEE}
\end{figure}

\section{Conclusions}
\label{sec:Conclusions}
This paper has studied energy-efficient multi-cell multigroup coordinated joint beamforming and antenna selection with antenna-specific maximum power constraints and user-specific QoS constraints. Two different approaches based on mixed-Boolean programming and sparsity were proposed to solve this challenging problem. The resulting mixed-Boolean nonconvex optimization problem was tackled by a continuous relaxation and the successive convex approximation, where the antennas for which continuous antenna selection variables converge to zero are switched off. Different convex and non-convex approximations were proposed to solve the problem with sparsity-based formulation. We have also considered the trade-off between energy efficiency and sum rate by proposing two approaches to solve the problem. The numerical results have illustrated that both the continuous relaxation of mixed-Boolean program and the sparsity-based approaches provide very good performance for the considered problems. Moreover, the proposed methods can significantly improve energy efficiency over the method without antenna selection by switching off a portion of the antennas but using larger per-antenna transmit power. It is also observed that joint beamforming and antenna selection can be used to achieve wider energy efficiency and sum rate trade-off curve. \textcolor{black}{It was
also observed that when the number of users is fixed, it
is energy-efficient to use a relatively small number of active
antennas if the number of available antennas is significantly
larger than the number of users in the system. Moreover,
increasing the number of antennas significantly for a fixed number
of users is not beneficial, because the energy-efficient number
of active antennas starts to saturate. The EE-SR trade-off results also showed that the JBAS scheme can maintain the same average rate as the beamforming only method with more than 25\% increase in the energy efficiency. This gain was achieved by switching off approximately half of the antennas. As a general
conclusion, the energy-efficient beamforming strategy is not to
use very low per-antenna power, but rather switch off many
antennas and design energy-efficient beamformers for these
and increase the per-antenna power.}

%\section{Acknowledgements}
%This work was supported in part by Infotech Oulu Doctoral Program and the Academy of Finland under projects MESIC belonging to the WiFIUS program with NSF, and WiConIE. It was also supported by a research grant from Science Foundation Ireland and is co-funded by the European Regional Development Fund under Grant 13/RC/2077, projects FNR SEMIGOD, SATSENT, INWIPNET and H2020 SANSA. The first author has been supported by KAUTE Foundation, HPY Research Foundation, Walter Ahlström Foundation, Tauno Tönning Foundation, and Nokia Foundation.

\bibliographystyle{IEEEtran}
%%%\footnotesize
%%% argument is your BibTeX string definitions and bibliography database(s)
%%%\bibliography{IEEEabrv,pimrc2013}
%\bibliography{IEEEabrv,referencesNam}


\appendices

\textcolor{black}{\section{Proof of Lemma \ref{Lemma1}}\label{App1}
Let use first prove inequality (i).
Assume $a_{b,i}^*,\hat{\mathbf{w}}_{b,i}^{*}$ is an optimal solution (for some ($b,i$)) of the continuous relaxation of \eqref{EEmax} (i.e., using $||\hat{\mathbf{w}}_{b,i}||_2^2 \leq a_{b,i}P_{\text{max}}$). Then, the denominator of the objective function \eqref{eq:EE:obj} at the optimal point is written as $\text{PC}_{\text{cont,orig}}=1/\eta\sum_{b\in\mathcal{B}}\sum_{i\in\mathcal{N}_b}||\hat{\mathbf{w}}_{b,i}^*||_2^2+\sum_{b\in\mathcal{B}}\sum_{i\in\mathcal{N}_b}a_{b,i}^*P_{\text{RF}}+P_0$. Note that in this case the constraint $||\hat{\mathbf{w}}_{b,i}||_2^2 \leq a_{b,i}P_{\text{max}}$ may not be tight. On the other hand, assume that we then use $||\hat{\mathbf{w}}_{b,i}||_2^2 \leq a_{b,i} v_{b,i}$, i.e., solve continuous relaxation of \eqref{EEmaxEqBinary}. Then, assume that $a_{b,i}^*,\hat{\mathbf{w}}_{b,i}^{*}$ (i.e., optimal solution of the relaxation of \eqref{EEmax}) is its optimal solution.
The power consumption can be written as %$PC_1=1/\eta\sum_{b\in\mathcal{B}}\sum_{i\in\mathcal{N}_b}\frac{||\hat{\mathbf{w}}_{b,i}||_2^2}{a_{b,i}^*}+\sum_{b\in\mathcal{B}}\sum_{i\in\mathcal{N}_b}a_{b,i}^*P_{\text{RF}}+P_0$
$\text{PC}_{\text{cont},\chi=1}=1/\eta\sum_{b\in\mathcal{B}}\sum_{i\in\mathcal{N}_b}v_{b,i}+\sum_{b\in\mathcal{B}}\sum_{i\in\mathcal{N}_b}a_{b,i}^{*}P_{\text{RF}}+P_0$. Now, assuming that $a_{b,i}^*,\hat{\mathbf{w}}_{b,i}^{*}$ would be the optimal solution, according to the constraint $\frac{||\hat{\mathbf{w}}_{b,i}^{*}||_2^2}{a_{b,i}^{*}} \leq v_{b,i}$, we can always decrease $v_{b,i}$ compared to the original formulation to make the denominator smaller (unless $||\hat{\mathbf{w}}_{b,i}||_2^2=P_{\text{max}}$) such that $v_{b,i}^{*} = \frac{||\hat{\mathbf{w}}_{b,i}^{*}||_2^2}{a_{b,i}^{*}}$. This also contradicts the fact that $a_{b,i}^*,\hat{\mathbf{w}}_{b,i}^{*}$ are optimal for the continuous relaxation of \eqref{EEmaxEqBinary}, and implies that the optimal objective value has to be smaller for the continuous relaxation of \eqref{EEmaxEqBinary}.
%, contradicting the fact that $v_{b,i}^{*}$ is the optimal solution.
Thus, because $a_{b,i}^{*}\leq 1$, it has to hold that $||\hat{\mathbf{w}}_{b,i}^{*}||_2^2\leq v_{b,i}^* =\frac{||\hat{\mathbf{w}}_{b,i}^{*}||_2^2}{a_{b,i}^{*}}$. This yields $\text{PC}_{\text{cont},\chi=1}\geq \text{PC}_{\text{cont,orig}}$, which again implies that $\text{EE}_{\text{cont},\chi=1} \leq \text{EE}_{\text{cont,orig}}$. Similarly the proof for (ii) follows from the fact that $a_{b,i}^\chi \leq a_{b,i}$, which implies that $\frac{||\hat{\mathbf{w}}_{b,i}^{*}||_2^2}{(a_{b,i}^{*})^\chi} \geq \frac{||\hat{\mathbf{w}}_{b,i}^{*}||_2^2}{a_{b,i}^{*}}$, i.e., $\text{PC}_{\text{cont},\chi=m}\geq \text{PC}_{\text{cont},\chi=1}$ and $\text{EE}_{\text{cont},\chi=m}\leq \text{EE}_{\text{cont},\chi=1}$. The inequality in (iii) is simply because the optimal value with any continuous relaxation is larger than that of the Boolean formulation.
}

\textcolor{black}{\section{Proof of Lemma \ref{Lemma2}}\label{App2}
For proving the lemma, we show that constraints \eqref{eq:EEmax:reform1:rate} and \eqref{eq:EEmax:reform1:weakestRATE}
are active at the optimality by  contradiction. Let ${\bf w}^{\ast},\boldsymbol\gamma^{\ast},\mathbf{v}^{\ast},\mathbf{a}^{\ast},\mathbf{r}^{\ast}$
be an optimal solution of \eqref{eq:EEmax:reform1} with optimal value $EE^\ast$ and suppose that \eqref{eq:EEmax:reform1:rate}
is not active at the optimum for the worst user in some group, i.e., $\gamma_k^\ast < \frac{|\mathbf{h}_{b_g,k}\herm \mathbf{w}_{g}^\ast|^{2}}{{N_0 + \sum\limits_{u \in \mathcal{G} \setminus \{g\}} |\mathbf{h}_{b_u,k}\herm \mathbf{w}_{u}^\ast|^{2}}}$ for some ${k\in\mathcal{K}_g}$. Then we can scale down the transmit power for user group
$g$ and achieve a new beamformer ${{\bf m}}_{g}$
such that $\|{{\bf m}}_{g}\|_{2}^{2}=\tau\|{\bf w}_{g}^\ast\|_{2}^{2}<\|{\bf w}_{g}^\ast\|_{2}^{2}$
for $\tau\in[0,1]$ while keeping the others  unchanged,
i.e ${{\bf m}}_{l}={\bf w}_{l}^{\ast}$ for all $l\neq g$.
In this way, we then achieve $\gamma_j^\ast < \frac{|\mathbf{h}_{b,j}\herm \mathbf{m}_{l}|^{2}}{{N_0 + \sum\limits_{u \in \mathcal{G} \setminus \{l\}} |\mathbf{h}_{b_u,j}\herm \mathbf{m}_{u}|^{2}}}$
for all $j\in\mathcal{K}\setminus \mathcal{K}_g$ since interference power at all the other users has reduced. Thus, to improve the objective, we could then increase $\gamma_j^\ast$ for all $j\in\mathcal{K}\setminus \mathcal{K}_g$ until $r_l=\min_{j\in\mathcal{K}_l}\log(1+\gamma_j), \forall l\in\mathcal{G}$.
In addition, we then also have $||\hat{\mathbf{m}}_{b_g,i}||_2^2 < a_{b_g,i}^\chi v_{b_g,i}, \forall i\in\mathcal{N}_{b_g}$ (because we have reduced the transmit power for group $g$), which means that we could reduce either $a_{b_g,i}$ or $v_{b_g,i}$ to make the denominator of \eqref{EqObj} smaller. This implies that we have $EE({\bf w}^{\ast}) < EE({\bf m})$.
Consequently, this contradicts the fact that ${\bf w}^{\ast},\boldsymbol\gamma^{\ast},\mathbf{v}^{\ast},\mathbf{a}^{\ast},\mathbf{r}^{\ast}$
is the optimal solution, completing the proof.}

{\color{black}\section{Convergence Analysis of the Iterative Algorithm for Solving \eqref{eq:EEmax:reform2}}\label{app:convergence}
Here we show that the the objective function \eqref{eq:EEmax:reform3:obj}  converges monotonically, i.e., improves at every iteration and is bounded above. Let $g_n$ be the optimal objective obtained at iteration $n$ of the proposed SCA-based algorithm, i.e., $g_n$ is the optimal objective of  \eqref{eq:EEmax:reform3}. We will show that $g_{n+1}\geq g_{n}$, i.e., the objective sequence is monotonically increasing. To this end we prove that the solution obtained at iteration $n$ is also feasible to the problem considered at iteration $n+1$.

Let us first focus on the constraint \eqref{eq:EEmax:reform3:ApproxB} and denote by $\mathbf{w}_g^{*(n)}$, $v_{b,i}^{*(n)}$, and $  a_{b,i}^{*(n)}$ the optimal values of $\mathbf{w}_{g}$, $v_{b,i}$ and  $a_{b,i}$, respectively at iteration $n$. It immediately holds that
\begin{equation}\label{eq:upperbound}
\frac{||\hat{\mathbf{w}}_{b,i}^{*(n)}||_2^2}{v_{b,i}^{*(n)}} \leq \Upsilon_{b,i}^{(n-1)}(a_{b,i}^{*(n)}) \leq  (a_{b,i}^{*(n)})^\chi
\end{equation}
where the first inequality is obvious and the second one is due to \eqref{eq:EEmax:Approx}.
At iteration $n+1$ \eqref{eq:EEmax:reform3:ApproxB} becomes
\begin{equation}
\frac{||\hat{\mathbf{w}}_{b,i}||_2^2}{v_{b,i}} \leq (1-\chi)(a_{b,i}^{*(n)})^\chi + \chi {(a_{b,i}^{*(n)})}^{(\chi-1)}a_{b,i}, \forall b \in \mathcal{B}, i \in \mathcal{N}_b
\end{equation}
Substituting $\mathbf{w}_g^{*(n)}$, $v_{b,i}^{*(n)}$, and $  a_{b,i}^{*(n)}$ into the above inequality results in
\begin{align}
\frac{||\hat{\mathbf{w}}_{b,i}^{*(n)}||_2^2}{v_{b,i}^{*(n)}} &\leq (1-\chi)(a_{b,i}^{*(n)})^\chi + \chi {(a_{b,i}^{*(n)})}^{(\chi-1)}a_{b,i}^{*(n)}\nonumber\\
&=(a_{b,i}^{*(n)})^\chi, \forall b \in \mathcal{B}, i \in \mathcal{N}_b
\end{align}
which is true due to \eqref{eq:upperbound}. That is, $\mathbf{w}_g^{*(n)}$, $v_{b,i}^{*(n)}$, and $  a_{b,i}^{*(n)}$ are feasible to  \eqref{eq:EEmax:reform3:ApproxB} at iteration $n+1$. Similarly, the same result can also proved for \eqref{eq:EEmax:reform3:approxquad}. Consequently, we can conclude that the  solution obtained at iteration $n$ is also feasible to the convex program considered at iteration $n+1$, and thus $g_{n+1}\geq g_{n}$.

The sequence $g_{n}$ is bounded from above due to the limited transmit power, and thus it is convergent.
It is difficult to comment on the convergence of the iterates (i.e., optimization variables) generated by the algorithm as the objective in \eqref{eq:EEmax:reform3:obj} is not strongly concave.
% However, if the iterates converge, it is possible to check that they converge to a stationary point of \eqref{eq:EEmax:reform2}.
A possible way to achieve the convergence of the iterates is to introduce a sufficiently large proximal term in the objective of \eqref{eq:EEmax:reform3:obj}. However, this method is quite involved and, thus, not considered in this paper.
}

\section{Iterative SOCP Method to Solve \eqref{EEmax2}}\label{App3}
The solution proposed in Section \ref{SolveContRelax} requires solving a generic non-linear concave-convex fractional program \eqref{eq:EEmax:reform3} in each iteration. This can be equivalently transformed to a generic non-linear convex program \eqref{eq:EEmax:reform4} which is still difficult to solve efficiently due to the exponential cone in constraint \eqref{eq:EEmax:reform4:rateconst}. Here we aim at finding a more efficient formulation. Specifically, all the other constraints in \eqref{eq:EEmax:reform3} admit the second-order cone form except the log-term in \eqref{eq:EEmax:reform1:weakestRATE}. To avoid the use of log-function, we need to find a concave lower bound for $\log(1+\gamma_k)$ to fulfill the conditions of the SCA.
%To this end, we can use different methods which can have different impacts on the convergence speed of the algorithms.
To this end, we use the following lower bound approximation for the concave log-function \cite{Papandriopoulos-09SCALE}
%\begin{equation}
%\log(1+\gamma_k) \geq \log(1+\gamma_k^{(n)}) + \frac{1}{1+\gamma_k^{(n)}}(\gamma_k-\gamma_k^{(n)})^2 \triangleq \Xi_k^{(n)}(\gamma_k)
%\end{equation}
\begin{equation}
\log(1+\gamma_k) \geq -\frac{\nu_{k,1}}{\gamma_k} + \nu_{k,2} \triangleq \Xi_k^{(n)}(\gamma_k)
\end{equation}
which is tight at $\gamma_k=\gamma_k^{(n)}$, when the coefficients $\nu_{k,1}, \nu_{k,2}$ are chosen as
\begin{equation}
\nu_{k,1} = \frac{(\gamma_k^{(n)})^2}{1+\gamma_k^{(n)}}, \nu_{k,2} = \log(1+\gamma_k^{(n)}) + \frac{\gamma_k^{(n)}}{1+\gamma_k^{(n)}}
\end{equation}
%Another approximation is to apply the method from \cite{Tervo-16Arxiv}, where
%\begin{equation} \label{approximation}
%\begin{aligned}
%\log(1+\gamma_k ) & \geq  \log(1+\gamma^{(n)}_{k})+\frac{1}{1+\gamma^{(n)}_{k}}(\gamma_k-\gamma^{(n)}_{k}) \nonumber \\ & -\frac{L}{2}(\gamma_k-\gamma_k^{(n)})^2 \triangleq \Xi_k^{(n)}(\gamma_k).
%\end{aligned}
%\end{equation}
%A third formulation could be
%\begin{equation} \label{approximation}
%\begin{aligned}
%\log(1+\gamma_k )  \geq  \log(1+\gamma^{(n)}_{k})+\frac{1}{1+\gamma^{(n)}_{k}}(\gamma_k-\gamma^{(n)}_{k}) -\frac{L}{2}(\gamma_k-\gamma_k^{(n)})^2 \triangleq \Xi_k^{(n)}(\gamma_k).
%\end{aligned}
%\end{equation}

As a result, we follow the description of Algorithm \ref{algo:iterative} but solve the following SOCP at step 2
\begin{subequations}
\label{eq:EEmax:SOCP}
%\begin{align}
\begin{eqnarray}
& \hspace{-20pt} \underset{\phi,\bar{\mathbf{w}},\bar{\boldsymbol\gamma},\bar{\mathbf{v}},\bar{\mathbf{a}},\bar{\boldsymbol\beta},\bar{\mathbf{r}},\bar{\boldsymbol\rho}}{\maxi} &   \sum_{g\in\mathcal{G}}\bar{r}_g\label{eq:EEmax:reform5:obj}\\
& \hspace{-60pt} \st  & \hspace{-35pt} \sum\limits_{b\in\mathcal{B}}\sum\limits_{i\in\mathcal{N}_b}(\frac{1}{\eta}\bar{v}_{b,i} + P_{\text{RF}}\bar{a}_{b,i}) + \phi P_{0} \leq 1\label{eq:EEmax:reform5:C1}  \\
%&  & \hspace{-25pt}  ||\bar{\hat{\mathbf{w}}}_{b,i}||_2^2 \leq \bar{a}_{b,i}\bar{v}_{b,i},\; \forall b \in \mathcal{B}, i \in \mathcal{N}_b \label{eq:EE:PC}\\
&  & \hspace{-35pt} \frac{||\bar{\hat{\mathbf{w}}}_{b,i}||_2^2}{\bar{v}_{b,i}} \leq \phi\Upsilon_{b,i}^{(n)}(\frac{\bar{a}_{b,i}}{\phi}), \forall b \in \mathcal{B}, i \in \mathcal{N}_b\label{eq:EEmax:reform5:C2}\\
&  & \hspace{-35pt} \bar{v}_{b,i} \leq \phi P_{\text{max}}, \forall b \in \mathcal{B}, i \in \mathcal{N}_b\label{eq:EEmax:reform5:C3} \\
%&  & \hspace{-25pt} \Phi_{b,i}^{(n)}(\bar{\hat{\mathbf{w}}}_{b,i},\phi) \geq \phi(\frac{\bar{a}_{b,i}}{\phi})^\alpha P_{\text{min}}, \forall b \in \mathcal{B}, i \in \mathcal{N}_b \label{eq:EEmax:reform2:MinPC} \\
%&  & \hspace{-25pt} \bar{v}_{b,i} \geq \bar{a}_{b,i} P_{\text{min}}, \forall b \in \mathcal{B}, i \in \mathcal{N}_b \label{eq:EEmax:reform2:vMinPC} \\
&  & \hspace{-35pt} 0 \leq \bar{a}_{b,i} \leq \phi, \forall b \in \mathcal{B}, i \in \mathcal{N}_b\label{eq:EEmax:reform5:C4}\\
&  & \hspace{-35pt} \bar{\gamma}_k \leq \phi\Psi_k^{(n)}(\frac{\bar{\mathbf{w}}_{g}}{\phi},\frac{\bar{\beta}_{k}}{\phi}), \forall k \in \mathcal{K}\label{eq:EEmax:reform5:C5}\\
&  & \hspace{-35pt} \bar{r}_g \geq \phi \max_{k\in\mathcal{K}_g}(\bar{R}_k), \forall g \in \mathcal{G}\label{eq:EEmax:reform5:C6} \\
&  & \hspace{-35pt} \phi\bar{\beta}_k \geq {\phi^2 N_0 + \sum\limits_{u \in \mathcal{G} \setminus \{g\}} |\mathbf{h}_{b_u,k}\herm \bar{\mathbf{w}}_{u}|^{2}}, \forall k \in \mathcal{K}\label{eq:EEmax:reform5:C7}\\
&  & \hspace{-35pt} \bar{r}_g \leq \phi \Xi_k^{(n)}(\frac{\bar{\gamma}_k}{\phi}), \forall g \in \mathcal{G}, k \in \mathcal{K}_g.\label{eq:EEmax:reform5:C8}
%&  & \hspace{-35pt} \phi \Phi_{b,i}^{(n)}(\frac{\hat{\mathbf{w}}_{b,i}}{\phi}) \geq (\bar{a}_{b,i}-\bar{\rho}_{b,i}) P_{\text{min}}, \forall b \in \mathcal{B}, i \in \mathcal{N}_b
\end{eqnarray}
\end{subequations}
In the above problem, all the other constraints are linear except \eqref{eq:EEmax:reform5:C2} and \eqref{eq:EEmax:reform5:C7} which can be expressed as $||\mathbf{y}_1||^2 \leq y_2y_3$, where $\mathbf{y}_1$ is some vector, and $y_2,y_3$ are scalars. These constraints are equivalently written in the SOC form as
\begin{equation}\label{generalform}
||\mathbf{y}_1\trans, 1/2(y_2 - y_3)||_2 \leq 1/2(y_2 + y_3).
\end{equation}
\textcolor{black}{More specifically, let us first consider \eqref{eq:EEmax:reform5:C2} and denote $x_{b,i}^{(n)}\triangleq (1-\chi)(a_{b,i}^{(n)})^\chi, z_{b,i}^{(n)}\triangleq \chi(a_{b,i}^{(n)})^{(\chi-1)}$ as the fixed coefficients in $\Upsilon_{b,i}^{(n)}(\frac{\bar{a}_{b,i}}{\phi})$. Then, we can write it equivalently as $||\bar{\hat{\mathbf{w}}}_{b,i}||_2^2\leq \bar{v}_{b,i}(\phi x^{(n)}+\bar{a}_{b,i}z^{(n)})$. Thus, according to \eqref{generalform}, we can write $\mathbf{y}_1=\bar{\hat{\mathbf{w}}}_{b,i}, y_2=v_{b,i}, y_3=\phi x_{b,i}^{(n)}+\bar{a}_{b,i}z_{b,i}^{(n)}$, and substituting these to \eqref{generalform} is equivalent to \eqref{eq:EEmax:reform5:C2}. On the other hand, \eqref{eq:EEmax:reform5:C7} is readily in a form $||\mathbf{y}_1||^2 \leq y_2y_3$, where $||\mathbf{y}_1||^2=\phi^2 N_0 + \sum\limits_{u \in \mathcal{G} \setminus \{g\}} |\mathbf{h}_{b_u,k}\herm \bar{\mathbf{w}}_{u}|^{2}, y_2=\phi, y_3=\bar{\beta}_k$. In this case, $\mathbf{y}_1=[\phi\sqrt{N_0}, I_{1,k},\ldots,I_{g-1,k},I_{g+1,k},\ldots,I_{G,k}]\trans$, where $I_{u,k}=\mathbf{h}_{b_u,k}\herm \bar{\mathbf{w}}_{u}$ is a scalar.}

\textcolor{black}{\section{Complexity Analysis}\label{App4}
Here we provide a complexity comparison based on the worst-case analysis presented in \cite{Ben-Tal-01}. The worst-case complexity of Alg. 1 depends on the number of variables and can be upper bounded as
%\begin{eqnarray}
%&  &
%Q_1\mathcal{O}(\sum_{b\in\mathcal{B}}2N_bG_b+\sum_{b\in\mathcal{B}}N_b+\sum_{b\in\mathcal{B}}N_b+K+K+G+1)^4\\
%&  & + Q_2\mathcal{O}(\sum_{b\in\mathcal{B}}2\bar{N}_bG_b+K+K+G+1)^4
%\end{eqnarray}
\begin{eqnarray}
&  &
Q_1\mathcal{O}(\sum\limits_{b\in\mathcal{B}}2N_bG_b+2\sum\limits_{b\in\mathcal{B}}N_b+2K+G+1)^4\\
&  & + Q_2\mathcal{O}(\sum\limits_{b\in\mathcal{B}}2\bar{N}_bG_b+2K+G+1)^4
\end{eqnarray}
where $Q_1$ and $Q_2$ are the number of performed iterations in the relaxed problem (steps 1-5) and lower-dimensional problem for fixed antenna set (steps 6-7), respectively, $\bar{N}_b$ is the number of antennas selected for transmission, $K$ is the total number of users, and $G$ is the total number of groups.
%The seven term represent the dimensions of $\bar{\mathbf{w}},\bar{\mathbf{a}},\bar{\mathbf{v}},\bar{\boldsymbol\gamma},\bar{\boldsymbol\beta},\bar{\mathbf{r}}$, and $\phi$, respectively.
Taking the dominant terms, it can be approximated as
\begin{equation}
Q_1\mathcal{O}(\sum_{b\in\mathcal{B}}N_bG_b+K)^4 + Q_2\mathcal{O}(\sum_{b\in\mathcal{B}}\bar{N}_bG_b+K)^4
\end{equation}
where the dominant term depends on the relation between $N_bG_b$ and $K$. Note that Alg. 1 `simple' reduces the complexity compared to Alg. 1 so that the second term in the above equation is ignored. The complexity reduction depends on the number of iterations and complexity to solve the lower-dimensional problem in steps 6-7 of Alg. 1.
We note that the above
upper bound for the complexity is quite conservative as a solution can be found
much faster in reality.
On the other hand, the worst-case complexity of solving the SOCP in Appendix \ref{App3} can be written as
\begin{eqnarray}
& & Q(\mathcal{O}(\sum_{b\in\mathcal{B}}N_bG_b+K)^{3}\\
&  & +\mathcal{O}((\sum_{b\in\mathcal{B}}N_bG_b+K)(\sum_{b\in\mathcal{B}}N_b+K)))
\end{eqnarray}
where $Q$ is the number of iterations. This algorithm provides a sharp complexity reduction compared to the generic formulation.}

\end{document}